\relax

\documentclass{article}
\usepackage{style/ijcai21}

\usepackage{times}
\usepackage{soul}
\usepackage{url}
\usepackage[hidelinks]{hyperref}
\usepackage[utf8]{inputenc}
\usepackage[small]{caption}
\usepackage{graphicx}
\usepackage{amsmath}
\usepackage{amsthm}
\usepackage{booktabs}
\usepackage{algorithm}
\usepackage{algorithmic}
\usepackage{amssymb}
\usepackage{enumitem}
\usepackage{subfigure}
\usepackage{amsfonts}       
\usepackage{nicefrac}       
\usepackage{subfiles}
\usepackage{verbatim}
\usepackage{tikz}
\usetikzlibrary{arrows}
\usetikzlibrary{decorations.pathreplacing}
\usetikzlibrary{shapes.multipart}
\newtheorem{theorem}{Theorem}

\newtheorem{lemma}{Lemma}

\newtheorem{problem}{Problem}

\theoremstyle{definition}

\newtheorem{claim}{Claim}

\graphicspath{{figures/}}
\usepackage{mathrsfs}
\usepackage{xspace}

\newcommand{\ex}[1]{ \mathbb{E} \left[ #1 \right] }
\newcommand{\prob}[1]{ P \left( #1 \right) }
\newcommand{\definef}{$f:2^U\to\mathbb{R}_{\geq 0}$\xspace}
\newcommand{\delt}[2]{\Delta f( #1 , #2 ) \xspace}
\DeclareMathOperator*{\argmax}{arg\,max}

\newcommand{\greedyratio}{$1-1/e$\xspace}

\newcommand{\pool}{\mathcal{S}}
\newcommand{\poolb}{$\mathcal{S}$\xspace}

\newcommand{\qiantime}{$\mathcal{O}(n\kappa^2)$\xspace}

\newcommand{\budget}{$\kappa$\xspace}
\newcommand{\sg}{\textsc{SG}\xspace}

\newcommand{\Sino}[1]{\pool[ #1 ]}
\newcommand{\cards}{\{0,...,P-1\}}
\newcommand{\bigO}[1]{\mathcal{O}( #1 )}
\newcommand{\instance}{SM$(f,\kappa)$\xspace}

\newcommand{\defineP}{$P\in\{1,...,n+1\}$\xspace}
\newcommand{\defineT}{$T\in \mathbb{Z}_{\geq 0}$\xspace}
\newcommand{\definep}{$p\in(0,1]$\xspace}

\newcommand{\definepool}{$\mathcal{S}\subseteq 2^U$\xspace}

\newcommand{\defineepsilon}{$\epsilon\in(0,1)$\xspace}
\newcommand{\definexi}{$\xi\in(0,1)$\xspace}

\newcommand{\ea}{\textsc{PO}\xspace}
\newcommand{\eal}{\textsc{Pareto Optimization}\xspace}
\newcommand{\ealb}{Pareto Optimization\xspace}

\newcommand{\mutate}{\textsc{Mutate}\xspace}

\newcommand{\earatio}{$(1-\epsilon)(1-1/e)$\xspace}
\newcommand{\eatime}{$\mathcal{O}(nP\ln(1/\epsilon))$\xspace}
\newcommand{\eaTa}{2enP}
\newcommand{\eaTb}{8enP\ln(1/\epsilon)}
\newcommand{\eaT}{\eaTb}
\newcommand{\betaf}{\max_{X\in\pool, |X|\leq\omega}f(X)}
\newcommand{\betasetline}{\argmax\{f(X): X\in\pool, |X|\leq\omega\}}

\newcommand{\maxf}[1]{\max \{f(X): X\in\pool, |X|\leq #1 \}}
\newcommand{\maxfb}[1]{\max_{X\in\pool, |X|\leq #1 }f(X)}
\newcommand{\argmaxf}[1]{\argmax \{f(X): X\in\pool, |X|\leq #1 \}}
\newcommand{\argmaxfb}[1]{\argmax_{X\in\pool, |X|\leq #1 } f(X)}
\newcommand{\eaalg}{Alg.\ \ref{algorithm:ea}}

\newcommand{\kbea}{$\kappa$-\textsc{BPO}\xspace}
\newcommand{\bea}{\textsc{BPO}\xspace}
\newcommand{\beal}{\textsc{Biased Pareto Optimization}\xspace}
\newcommand{\bealb}{Biased Pareto Optimization\xspace}
\newcommand{\selectbea}{\textsc{Select-BPO}\xspace}
\newcommand{\flip}{\textsc{Flip-Coin}\xspace}
\newcommand{\uniform}{\textsc{Uniform-Random}\xspace}
\newcommand{\bearatio}{$(1-\epsilon)(1-1/e-\epsilon)$\xspace}
\newcommand{\beatime}{$\mathcal{O}(n\ln(P)\ln(1/\epsilon))$\xspace}

\newcommand{\beaTa}{2ne\ln(1/\epsilon)\maxpointer/p}
\newcommand{\beaTb}{8\ln(n)\maxpointer/p}

\newcommand{\beaprob}{$p/\maxpointer$}

\newcommand{\beaalg}{Alg.\ \ref{algorithm:bea}}
\newcommand{\kappabea}{$\kappa$-\bea}
\newcommand{\kappabearatio}{\bearatio}
\newcommand{\kappabeatime}{$\mathcal{O}(n\ln(1/\epsilon))$\xspace}

\newcommand{\defineH}{$H=en\ln(1/\epsilon)/\kappa$\xspace}

\newcommand{\Hj}{$e\ln(1/\epsilon)/\xi^j$\xspace}

\newcommand{\maxpointer}{\lceil\ln(P)/\ln\left(1/\xi\right)\rceil}

\newcommand{\qiansmt}{\citeauthor{qian2015subset} \shortcite{qian2015subset}\xspace}
\newcommand{\qiansma}{\citeauthor{qian2015subset}\xspace}

\newcommand{\friedt}{\citeauthor{Friedrich2014} \shortcite{Friedrich2014}\xspace}
\newcommand{\frieda}{\citeauthor{Friedrich2014}\xspace}

\newcommand{\mirza}{\cite{mirzasoleiman2015lazier}\xspace}

\usepackage{tikz}
\usetikzlibrary{arrows}
\usepackage{xcolor}

\title{Faster Guarantees of Evolutionary Algorithms for Maximization of Monotone Submodular Functions}

\author{
    Victoria G. Crawford
    \affiliations
    University of Florida
    \emails
    vcrawford01@ufl.edu
}

\begin{document}

\maketitle

\begin{abstract}
  In this paper, the monotone submodular maximization problem (SM)
  is studied. SM is to find a subset of size $\kappa$ from a universe of size $n$
  that maximizes a monotone submodular objective function $f$.
  We show using a novel analysis that the Pareto optimization algorithm achieves
  a worst-case ratio of \earatio in expectation for every cardinality
  constraint $\kappa < P$, where $P\leq n+1$ is an input, in \eatime queries of $f$.
  In addition, a novel evolutionary algorithm called the biased Pareto optimization
  algorithm,
  is proposed that achieves a worst-case ratio of \bearatio
  in expectation for every cardinality constraint $\kappa < P$ in \beatime
  queries of $f$. Further, the biased Pareto optimization algorithm can be modified in
  order to achieve a worst-case ratio of \bearatio in expectation for cardinality
  constraint $\kappa$ in \kappabeatime queries of $f$.
  An empirical evaluation corroborates our theoretical analysis of the
  algorithms, as the algorithms exceed the stochastic greedy solution value at roughly
  when one would expect based upon our analysis.
\end{abstract}

\section{Introduction}
\label{section:introduction}
A function $f:2^U\to\mathbb{R}_{\geq 0}$ defined on subsets of a ground set $U$
of size $n$
is monotone submodular if it possesses the following two properties:
(i) For all $A\subseteq B\subseteq U$, $f(A)\leq f(B)$ (monotonicity);
(ii) For all $A\subseteq B\subseteq U$
and $x\notin B$, $f(A\cup\{x\})-f(A) \geq f(B\cup\{x\})-f(B)$ (submodularity).
Monotone submodular set functions are found in many applications in machine learning and data mining.
Applications of SM include influence in social networks
\cite{kempe2003maximizing},
data summarization \cite{mirzasoleiman2013distributed},
dictionary selection \cite{Das2011},
and monitor placement \cite{Soma2015a}.
As a result, there has been much recent interest in optimization problems involving monotone
submodular functions.
One such optimization problem is the NP-hard Submodular Maximization Problem (SM),
defined as follows.

\begin{problem}[Submodular Maximization Problem (SM)]
  Let \definef be a monotone submodular function
  defined on subsets of the ground set $U$ of size $n$,
  and $f(\emptyset)=0$.
  Given a budget $\kappa\in [0,n]$, SM is to find
  $\text{argmax}_{|X|\leq\kappa}f(X).$
\end{problem}
An instance of SM is referred to as \instance.
It is assumed that the function $f$ is provided as a value oracle,
which when queried with a set $X$
returns the value of $f(X)$.
Time is measured in queries of $f$, as is the convention in submodular optimization \cite{Badanidiyuru2014}.

To approximate SM, the standard greedy algorithm is very effective.
\citeauthor{Nemhauser1978} \shortcite{Nemhauser1978} showed that the standard greedy algorithm
achieves the best ratio of $(1 - 1/e)$ for SM in
$\mathcal O(nk)$ queries to $f$.
In addition, faster versions of the greedy algorithm have been developed for SM
\cite{Badanidiyuru2014,mirzasoleiman2015lazier}.
In particular,
the stochastic greedy algorithm (SG) of \citeauthor{mirzasoleiman2015lazier}
\shortcite{mirzasoleiman2015lazier}
achieves ratio $1 - 1/e - \epsilon$ in expectation in $O( n \ln (1 / \epsilon ))$
queries to $f$.

Alternatively, one may take the Pareto optimization approach to SM:
Instead of maximizing $f$ for a cardinality constraint $\kappa$, SM is re-formulated as a
bi-objective optimization problem where the goal is to both maximize $f$
as well as minimize cardinality.
Instead of a single solution,
we seek a pool of
solutions none of which dominate another\footnote{In this
context, a solution $Y$ dominates $X$ if $f(X)\leq f(Y)$, $|X|\geq |Y|$,
and at least one of the two inequalities is strict.}.
Greedy algorithms can be used to develop such a pool, however
previous works \cite{Friedrich2014,qian2015subset} have employed bi-objective evolutionary
algorithms because they
iteratively improve the entire pool of solutions and can be run indefinitely.
The evolutionary algorithm \eal (\ea) has previously been shown
to find a \greedyratio approximate solution to \instance for all $\kappa < P$,
where $P\leq n+1$ is an input,
in expected $\mathcal{O}(nP^2)$ queries to $f$ \cite{Friedrich2014}.
Further, \ea has been demonstrated to make significant
empirical improvements over the standard greedy algorithms for SM \cite{qian2015subset}.
But as the size of data has grown
exponentially in recent times, a query complexity that is cubic in $n$ (for $P = \Omega(n)$)
makes these evolutionary algorithms a less attractive option.

\subsection{Contributions}
In this work, a novel analysis is provided for the algorithm \ea, and it is proven
that \ea achieves a worst-case ratio of \earatio in expectation for every instance \instance
with $\kappa < P$, where $P\leq n+1$ is an input,
in \eatime queries of $f$.
This removes a factor of $P$ from the query complexity of \friedt.
This novel analysis has potential to improve the
query complexity of other problems in monotone submodular optimization beyond SM
\cite{qian2015constrained,Qian2017,crawford19,bian2020efficient}.
This result is proven in Theorem \ref{theorem:ea}.

Next, a novel algorithm \beal (\bea) is proposed that is a similar in spirit but
faster version of \ea
for SM. It is proven that \bea achieves a worst-case ratio of \bearatio in expectation
for every instance \instance with $\kappa < P$ in \beatime queries of $f$.
This result is proven in Theorem \ref{theorem:bea}.
Further, a version of \bea for a specific cardinality constraint \budget,
$\kappa$-\beal (\kappabea), is proven to achieve a worst-case ratio of
\kappabearatio in expectation for instance \instance in \kappabeatime queries of $f$.
This result is proven in Theorem \ref{theorem:kappabea}.
This new algorithm
\kappabea thus matches the optimal SG algorithm in
terms of both approximation ratio and query complexity, while
maintaining the ability of \ea to continuously improve a pool of solutions.

The above theoretical results all extend to the more general setting of
monotone $\gamma$-weakly submodular\footnote{
A function \definef is $\gamma$-weakly submodular if for all
$X\subseteq Y \subseteq U$, and $u\notin Y$,
$\sum_{u\in Y\setminus X}\Delta f(X,u) \geq \gamma \left(f(X\cup Y)-f(X)\right)$.
If $\gamma = 1$, then $f$ is submodular.} functions \cite{Das2011},
but with different approximation guarantees that depend on $\gamma$.

An empirical evaluation corroborates our theoretical analysis of the
algorithms, as the algorithms exceed the SG solution value at roughly
when one would expect based upon our analysis.


\subsection{Additional Related Work}
\label{section:relatedwork}
Evolutionary algorithms have been studied for many combinatorial optimization problems
\cite{laumanns2002running,neumann2007randomized,friedrich2010approximating}.
In particular, evolutionary algorithms have been analyzed for problems in submodular optimization
including
SM \cite{Friedrich2014,qian2015subset,roostapour2019pareto},
submodular cover \cite{qian2015constrained,crawford19},
SM with more general cost constraints \cite{bian2020efficient},
and noisy versions of SM \cite{Qian2017}.


\friedt studied a slight variant of \ea where the pool \poolb is
initialized to contain a random set,
and $P=n$.
\frieda proved that their variant of \ea finds a \greedyratio approximate
solution to \instance
in expected $\mathcal{O}(n^2\ln(n) + n^2 \kappa)$ queries of $f$.
It is easy to modify their analysis to see that \ea finds a \greedyratio approximate
solution to \instance for all $\kappa < P$ in expected $\mathcal{O}(nP^2)$ queries of $f$.
The argument of \ea used in the proof of Theorem \ref{theorem:ea} of Section \ref{section:ea}
is substantially different compared to the argument of \frieda
because it analyzes the
expected time until an \textit{expected} approximation ratio is analyzed,
resulting in a speedup to $\mathcal{O}(nP)$ queries of $f$.
In addition, the result of Theorem \ref{theorem:ea}
is in deterministic time due to an application of the Chernoff bound.

\qiansmt considered the subset selection problem, which is a special case of
the monotone $\gamma$-weakly submodular maximization
problem. \qiansma fixed $P=2\kappa$, and showed that for the cardinality constraint
$\kappa$ \ea finds a $1-e^{-\gamma}$ approximate solution
in expected \qiantime queries of $f$.
Their results can be generalized beyond subset selection to the monotone
$\gamma$-weakly submodular maximization
problem with cardinality constraint \budget.

The algorithm \bea, presented in Section \ref{section:bea},
uses a novel,
biased selection procedure to identify sets for mutation.
Because of the biased selection procedure, \bea is the first evolutionary algorithm that
has an approximation guarantee in nearly linear queries of $f$ close to that of the greedy algorithm for SM.



%


\section{Algorithms and Theoretical Results}
\label{section:theoretical}
The theoretical contributions of the paper are presented in this section.
In particular, a new theoretical analysis of the algorithm \eal (\ea) is presented for
SM in Section \ref{section:ea},
the novel algorithm \beal (\bea) is presented and analyzed for SM in Section
\ref{section:bea}, and the faster modification of \bea for a specific cardinality
constraint, $\kappa$-\beal (\kappabea), is presented and analyzed for SM in
Section \ref{section:kappabea}.
The full version of the paper includes an appendix where additional theoretical
details from Section \ref{section:theoretical} are filled in.

\paragraph{Definitions and Notation}
The following notation and definitions will be used throughout Section \ref{section:theoretical}.
Let \definef, $X\subseteq U$, and $x\in U$.
(i) Marginal gain: $\delt{X}{x} = f(X\cup\{x\})-f(X)$.
(ii) The membership of $x$ is \textit{flipped in} $X$
    means that if $x\in X$, then $x$ is
    removed from $X$; and if $x\notin X$,
    then $x$ is added to $X$.
(iii) If $\nu$ is a random variable, then $\ex{\nu}$ denotes the expected value of $\nu$. If $A$ is a random event, then $\prob{A}$ denotes the probability of $A$ occurring.
(iv) Let $\pool\subseteq 2^U$. Then if there exists a unique $Y\in\pool$ such that
  $|Y|=i$, define $\Sino{i}=Y$. If no such $Y$ exists, or there are multiple elements of
  \poolb of cardinality $i$, $\Sino{i}$ is undefined.
(v) Let $X,Y \subseteq U$. Then $X\preceq Y$ if $f(Y)\geq f(X)$ and
  $|Y|\leq|X|$. If at least one of the two inequalities is strict, then
  $X\prec Y$ and $Y$ \textit{dominates} $X$. If
  $f(Y)=f(X)$ and $|Y|=|X|$ then it is said that $X$ is \textit{equivalent} to $Y$.



\subsection{\ealb (\ea)}
\label{section:ea}
In this section, it is proven that in time \eatime, \ea produces
a \earatio approximate solution in expectation for
every cardinality constraint $\kappa < P$, where $P\leq n+1$ is an input.
If $P=\bigO{\kappa}$ for a fixed cardinality constraint $\kappa$, then
\ea produces solutions for SM with similar theoretical guarantees
to that of the standard greedy algorithm in the same asymptotic
time, which shows the practicality of evolutionary algorithms such as \ea for
SM.

\subsubsection{Description of \ea}
\label{section:eaoverview}

\begin{algorithm}[t]
   \caption{\ea$(f, P, T)$: \eal (\ea)}
   \label{algorithm:ea}
   \begin{algorithmic}[1]
     \STATE {\bfseries Input:} \definef; \defineP; \defineT.
     \STATE {\bfseries Output:} \definepool.
     \STATE $\pool \gets \{\emptyset\}$
     \FOR {$t\gets 1$ to $T$} \label{line:for}
      \STATE $i \gets \uniform(\{0,...,P-1\})$ \label{line:selectea}
      \IF{$\Sino{i}$ exists}
       \STATE $B\gets\Sino{i}$
       \STATE $B' \gets $ \mutate($B$) \label{line:mutate}
       \IF {$|B'|<P$ and $\nexists Y \in \pool$ such that $B' \preceq Y$}
       \STATE $\pool \gets \pool \cup \{B'\}\setminus \{Y\in\pool: Y\prec B'\}$ \label{line:add}
       \ENDIF
      \ENDIF
     \ENDFOR
     \STATE \textbf{return} $\pool$
\end{algorithmic}
\end{algorithm}

\noindent In this section, \ea (\eaalg) is described.
The set \definepool is referred to as \textit{the pool}, and each iteration of
the \textbf{for} loop is referred to as \textit{an iteration}.
The pool initially contains only the empty set; its maximum size is determined
by input parameter $P$.
During each iteration, $i\in\cards$ is chosen uniformly randomly
(Line \ref{line:selectea} of \eaalg), and if $B=\Sino{i}$ exists then
it is selected from \poolb to be mutated, otherwise \ea continues to the
next iteration. The subroutine
\mutate takes $B$ and randomly mutates it into $B'\subseteq U$ as follows:
for each $u\in U$, flip the membership of $u$ in $B$ with probability $1/n$.
Finally, if no set in the pool dominates or is equivalent to $B'$ and $|B'| < P$,
then $B'$ is added to the pool and all sets that $B'$
dominates are removed.
There is at most one new query of $f$ on each iteration of \ea, and therefore
the input $T$ is equal to the query complexity.
Pseudocode for the subroutine \mutate is provided in the appendix.

\subsubsection{Analysis of \ea for SM}
\label{section:eaanalysis}
In this section,
the approximation result of the algorithm \ea for SM is presented.
Omitted proofs are given in the appendix.
The statement of Theorem \ref{theorem:ea} easily generalizes to $\gamma$-weakly
submodular objectives $f$, where $1-1/e$ is replaced with $1-1/e^{\gamma}$.

\begin{theorem}
  \label{theorem:ea}
  Suppose \ea is run with input $(f, P, T)$, where \definef is monotone submodular,
  \defineP, and $T> \eaT$. Let
  \poolb be the pool of \ea at the end of iteration $T$.
  Then for any $\kappa < P$,
  $$\ex{\max_{X\in\mathcal{S}, |X|\leq \kappa}f(X)}
   \geq (1-\epsilon)(1-1/e)\max_{|X|\leq \kappa}f(X).$$
\end{theorem}

\paragraph{Overview of Proof of Theorem \ref{theorem:ea}}
Given \instance with optimal solution $A^*$,
the standard greedy algorithm iteratively picks into its solution
the element in $U$ of highest marginal gain until \budget elements have been picked.
Existing analyses of \ea for SM \cite{Friedrich2014,qian2015subset}
analyze the time it takes until essentially the standard greedy algorithm randomly occurs within \ea.
If instead of iteratively picking the element of highest marginal gain, a uniformly
random element of $A^*$ (possibly already chosen) is picked into the solution
until \budget elements are picked (this could be viewed as an idealized version of
the stochastic greedy algorithm of \citeauthor{mirzasoleiman2015lazier} \shortcite{mirzasoleiman2015lazier}) then
the same approximation guarantee as the standard greedy algorithm (\greedyratio)
is achieved in expectation.
In the proof of Theorem \ref{theorem:ea}, the expected time until this second
algorithm randomly occurs within \ea is analyzed, which is a factor of $P$ faster.

\begin{proof}
Line numbers referenced are those in Algorithm \ref{algorithm:ea}.
Throughout the proof of Theorem \ref{theorem:ea}, the probability space
of all possible runs of \ea with the stated inputs is considered.
Let $\kappa < P$, and let $A^*=\argmax_{|X|\leq \kappa}f(X)$.
We may assume that $|A^*|=\kappa$, since $f$ is monotone.
In the proof of Theorem \ref{theorem:ea}, the random variable $\omega$ will be used,
defined inductively as follows:
\begin{itemize}[noitemsep]
  \item[(i)] Before the first iteration of \ea, $\omega$ is set to $0$.
  \item[(ii)] If the following
   two conditions are met,
    $\omega$ is incremented at the end of an iteration:
    1) $B=\betasetline$ is selected on Line \ref{line:selectea};
    and 2) \mutate on Line \ref{line:mutate} results in
    the membership of a single element $a^*\in A^*$ being flipped
    (\textit{i.e.}, it is either the case that \mutate returns $B'=B\cup\{a^*\}$ for $a^*\in A^*\setminus B$
    or $B'=B\setminus\{a^*\}$ for $a^*\in A^*\cap B$).
\end{itemize}
Intuitively, $\omega$ is used to track a solution within
$\mathcal{S}$ that has a high $f$ value relative to its cardinality.
In particular, the following lemma describes the
key property of $\omega$.
\begin{lemma}
  \label{lemma:expectation}
  At the end of every iteration of \ea
  $$\ex{\betaf} \geq \left(1-\left(1-\frac{1}{|A^*|}\right)^{\omega'}\right)f(A^*)$$
  where $\omega'=\min\{\omega,P-1\}$.
\end{lemma}
\noindent A further key point is that once a solution appears in \ea, i.e., it is returned by
\mutate, there always exists at least as good of a solution within \poolb.
\begin{lemma}
  \label{lemma:replace}
  Let $Y\subseteq U$ and $|Y|\leq a < P$. If $Y$ is returned by \mutate during
  iteration $i$ of \ea, then at the end of any iteration $j\geq i$ it holds that
  $\maxf{a} \geq f(Y).$
\end{lemma}

  \noindent Let event $F$ be that at the completion of a run of \ea, $\omega\geq \kappa$.
  Then it follows from Lemmas \ref{lemma:expectation} and \ref{lemma:replace}
  that
  \begin{align}
    \ex{\max_{X\in\mathcal{S}, |X|\leq \kappa}f(X)|F}
    \geq \left(1-\frac{1}{e}\right)f(A^*) \label{eqn:hfd66672}
  \end{align}
  Then the remainder of the
  proof of Theorem \ref{theorem:ea} is to deal with the probability that $\omega$
  reaches $\kappa$.
  To this end, the following
  lemma states that the run of \ea may be interpreted as a Bernoulli process.
  \begin{lemma}
    \label{lemma:bernoulli}
    Consider a run of \ea as a series of Bernoulli trials $Y_1,...,Y_T$,
    where each iteration is a
    trial and a success is defined to be when $\omega$ is incremented.
    Then $Y_1,...,Y_T$ are independent, identically distributed Bernoulli trials where the
    probability of success is
    \begin{align*}
      \frac{1}{P}\sum_{x\in A^*}\left(1-\frac{1}{n}\right)^{n-1}\frac{1}{n} \geq \frac{|A^*|}{enP}.
    \end{align*}
  \end{lemma}
  \noindent Then Lemma \ref{lemma:bernoulli} and the Chernoff bound can be used to prove that the
  probability of $\omega$ not reaching $\kappa$ after $T \geq \eaT$
  iterations of \ea is small. This is stated in the following lemma.

  \begin{lemma}
    \label{lemma:probability}
    $P\left(\sum_{i=1}^TY_i < \kappa\right) \leq \epsilon.$
  \end{lemma}
  \noindent Finally, Theorem \ref{theorem:ea} follows from
  the law of total expectation, Inequality \ref{eqn:hfd66672} and
  Lemma \ref{lemma:probability}.
\end{proof}

%



%


\subsection{\bealb (\bea)}
\label{section:bea}
\beal (\bea) is a novel evolutionary algorithm with nearly the same approximation results as \ea for SM in faster time.
Specifically, it is proven that in time \beatime,
\bea finds a \bearatio-approximate solution in expectation for
every cardinality constraint $\kappa < P$, where $P\leq n+1$ is an input.
Thus, \bea is faster than \ea by a factor of $\Omega(P/\ln(P))$;
further, it works similarly to \ea
but has a biased selection procedure instead of choosing uniformly
randomly.

\subsubsection{Description of \bea}
\label{section:beaoverview}
\begin{algorithm}[t]
   \caption{\bea$(f, P, T, p, \epsilon, \xi)$: \beal}
   \label{algorithm:bea}
   \begin{algorithmic}[1]
     \STATE {\bfseries Input:} \definef; \defineP; \defineT; \definep;
     \defineepsilon; \definexi.
     \STATE {\bfseries Output:} \definepool.
     \STATE $M\gets\maxpointer$
     \STATE $\beta_i\gets 0$, $\ell_i\gets 0$, $H_i \gets e\ln(1/\epsilon)/\xi^i \: \forall i \in\{1,...,M\}$
     \STATE $\pool \gets \{\emptyset\}$
     \FOR {$t\gets 1$ to $T$} \label{line:loopbea}
      \STATE $i\gets\uniform(\{0,...,P-1\})$ \label{line:random2}
      \IF {$\flip(p)=$ heads}
       \STATE $j\gets\uniform(\{1,...,M\})$ \label{line:random}
       \STATE $i\gets |\argmaxf{\beta_j}|$ \label{line:seti}
       \STATE $\ell_j \gets \ell_j + 1$ \label{line:incell}
       \IF{$\ell_j = H_j$}
        \STATE $\ell_j\gets 0$ \label{line:ell}
        \STATE $\beta_j\gets\beta_j+1$
       \ENDIF
      \ENDIF
      \IF{$\Sino{i}$ exists}
       \STATE $B\gets \Sino{i}$ \label{line:selectbea}
       \STATE $B' \gets $ \mutate($B$) \label{line:mutatebea}
       \IF {$|B'|<P$ and $\nexists Y \in \pool$ such that $B' \preceq Y$}
        \STATE $\pool \gets \pool \cup \{B'\}\setminus \{Y\in\pool: Y\prec B'\}$ \label{line:addbea}
       \ENDIF
      \ENDIF
     \ENDFOR
     \STATE \textbf{return} $\mathcal{S}$
\end{algorithmic}
\end{algorithm}

\noindent In this section, \bea (\beaalg) is presented.
Pseudocode for \bea can be found in Alg. \ref{algorithm:bea}.
In overview,
\bea follows a similar iterative procedure to \ea:
every iteration of the \textbf{for} loop,
a set in the $\pool$ is chosen for mutation;
and only sets that are not dominated by any
others are kept. The difference from \ea is
in the selection of the set for mutation;
a certain subset of sets in $\pool$ are
selected more frequently than others,
as determined by the parameters
\definep, \defineepsilon, and \definexi,
and the variables
$\beta_j$ for $j\in\{1,...,M\}$, where $M=\maxpointer$.
There is at most one new query of $f$ on each iteration of \bea, and therefore
the input $T$ is equal to the query complexity.
Next, the selection process is described in detail.

\paragraph{Selection process} During each iteration,
with probability $p$ \bea chooses $j$ from $\{1,...,M\}$ uniformly randomly
(Line \ref{line:random}) and then sets $i=|\argmaxf{\beta_j}|$
(Line \ref{line:seti}).
Otherwise $i$ is chosen uniformly randomly
from $\cards$ (Line \ref{line:random2}).
If $B=\Sino{i}$ exists then
it is selected from \poolb to be mutated, otherwise \bea continues to the
next iteration.
Initially, $\beta_j=0$ $\forall j \in\{1,...,M\}$.
$\beta_j$ is incremented to $\beta_j+1$ if
on $H_j=$\Hj iterations since the last increment of $\beta_j$
$j$ was chosen on Line \ref{line:random}.
The variable $\ell_j$ is used to determine when $\beta_j$ should be incremented:
$\beta_j$ is incremented during an iteration if and only if $\ell_j$ is set to 0
on Line \ref{line:ell}.
Notice that if $p=0$, \bea is equivalent to \ea.

\subsubsection{Analysis of \bea for SM}
\label{section:beaanalysis}
The approximation results of \bea for SM are now presented.
Lemmas referenced in the proof of Theorem \ref{theorem:bea} can be found in
the appendix.
The statement of Theorem \ref{theorem:bea} easily generalizes to $\gamma$-weakly
submodular objectives $f$, where $1-1/e-\epsilon$ is replaced with $1-1/e^{\gamma}-\epsilon$.

\begin{theorem}
  \label{theorem:bea}
  Suppose \bea is run with input $(f,P,T,p,\epsilon,\xi)$
  where \definef is monotone submodular, \defineP,
  $T\geq\max\{\alpha n\maxpointer, \beta\ln(n)\maxpointer\}$, where $\alpha=2e\ln(1/\epsilon)/p$
  and $\beta=8/p$,
  \definep, \defineepsilon, and \definexi.
  Let \poolb be the pool of \bea at the end of iteration $T$.
  Then for any $\kappa < P$,
  $$\ex{\max_{X\in\mathcal{S}, |X|\leq \kappa}f(X)}
   \geq (1-\epsilon)(1-1/e-\epsilon)\max_{|X|\leq \kappa}f(X).$$
\end{theorem}

\paragraph{Overview of Proof of Theorem \ref{theorem:bea}}
Consider \instance with optimal solution $A^*$.
Recall that in the proof of Theorem \ref{theorem:ea}
in Section \ref{section:ea}, the approximation ratio for \instance was proven by
analyzing the expected time until a variable $\omega$ reaches $\kappa$.
In order for $\omega$ to be incremented during an iteration,
$|\argmaxf{\omega}|$ must be selected on Line \ref{line:selectea}, which occurs
with probability $1/P$.
If we instead consider an alternative version of \ea where the selection is biased
towards choosing $|\argmaxf{\omega}|$ with constant probability $\alpha > 1/P$,
then $\omega$ reaches $\kappa$ faster.
The difficulty is that the value of $\omega$ is unknown, since it depends on
\mutate flipping the membership of an $a^*\in A^*$ and nothing else.
The idea behind \bea is that we can \textit{approximately} track $\omega$, and therefore
bias the selection.
In particular, for each \instance with $\kappa <P$, there exists
a $\beta_i$
that is approximately equal to the corresponding $\omega$ for \budget.

\begin{proof}
  Proofs of lemmas used can be found in the appendix.
  Lines numbers referenced are those in Algorithm \ref{algorithm:bea}.
  Throughout the proof of Theorem \ref{theorem:bea}, the probability space
  of all possible runs of \bea with the stated inputs is considered. An iteration of the
  for loop in \bea is simply referred to as an iteration.

  Consider any $\kappa < P$. Define $A^*=\argmax_{|X|\leq \kappa}f(X)$. Without loss of
  generality we may assume that $|A^*|=\kappa$, since $f$ is monotone.
  There exists $q\in\{1,...,\maxpointer\}$ such that
  \begin{align}
    \xi^qP < |A^*| \leq \xi^{q-1}P. \label{eqn:boundopt}
  \end{align}
  Then define $\omega=\beta_q$.

  The $\omega$ defined here serves a similar purpose to that defined in the
  proof of Theorem \ref{theorem:ea}; To track a solution within
  $\mathcal{S}$ that has a high $f$ value relative to its cardinality, as described in
  the following lemma.
  \begin{lemma}
    \label{lemma:expectationbea}
    At the end of every iteration of \bea
    $$\ex{\maxfb{\omega}} \geq \left(1-\left(1-\frac{1-\epsilon}{|A^*|}\right)^{\omega'}\right)f(A^*)$$
    where $\omega'=\min\{\omega,P-1\}$.
    \end{lemma}
  In addition, the property of
  \ea detailed in Lemma \ref{lemma:replace} of Theorem \ref{theorem:ea}
  clearly also holds for \bea.

  Define the event $F$ to be that at the
  completion of a run of \bea $\ell_q$ has been incremented (Line \ref{line:incell}
  of Algorithm \ref{algorithm:bea})
  $H_q\kappa$ times.
  If $\ell_q$ has been incremented $H_q\kappa$ times, then one may see that $\omega$ has been incremented
  $\kappa$ times. Once $\omega$ reaches $\kappa$, it clearly follows
  from Lemmas \ref{lemma:expectationbea} and \ref{lemma:replace} that
  \begin{align}
    \ex{f(A)|F} \geq \left(1-\frac{1}{e}-\epsilon\right)f(A^*) \label{eqn:hfd666}
  \end{align}
  where $A=\text{argmax}_{X\in\mathcal{S}, |X|\leq\kappa}f(X)$.

  We now analyze the probability that $\ell_q$ has been incremented $H_q\kappa$ times.
  To this end, we have the following lemma.
  \begin{lemma}
    \label{lemma:bernoullibea}
    Consider a run of \bea as a series of Bernoulli trials $Y_1,...,Y_T$,
    where each iteration is a
    trial and a success is defined to be when $\ell_q$ is incremented.
    Then $Y_1,...,Y_T$ are independent, identically distributed Bernoulli trials where the
    probability of success is \beaprob.
  \end{lemma}

  Finally, an analogous argument to that of Theorem \ref{theorem:ea} can be used to
  complete the proof of Theorem \ref{theorem:bea}.
  In particular, we bound the probability of event $F$ not occurring after
  $T\geq\max\{\alpha n\maxpointer, \beta\ln(n)\maxpointer\}$, where $\alpha=2e\ln(1/\epsilon)/p$
  and $\beta=8/p$,
  iterations of \bea by $\epsilon$ using the Chernoff bound, and then apply the law of total
  expectation.
  The details of the argument can be found in the appendix.
\end{proof}

\subsection{$\kappa$-\bealb ($\kappa$-\bea)}
\label{section:kappabea}
If a specific cardinality constraint
\budget is provided, a modified version of \bea, $\kappa$-\beal ($\kappa$-\bea),
can produce an approximate solution in expectation even faster than \bea.
In this section, the algorithm \kappabea is described, and it is proven that
\kappabea finds a \kappabearatio-approximate solution to \instance
in \kappabeatime queries of $f$.

\subsubsection{Description of \kappabea}
Pseudocode for \kappabea can be found in the appendix.
\kappabea is similar to \bea except \kappabea is only biased towards
picking a single element of \poolb, determined by the variable $\beta$.
The input parameters of \kappabea are the same as \bea except
\definexi is not needed.

During each iteration,
with probability $p$ \kappabea sets $i=|\argmaxf{\beta}|$.
Otherwise $i$ is chosen uniformly randomly
from $\cards$.
If $B=\Sino{i}$ exists then
it is selected from \poolb to be mutated, otherwise \kappabea continues to the
next iteration.
Initially, $\beta=0$.
$\beta$ is incremented to $\beta+1$ if
on \defineH iterations since the last increment of $\beta$,
$i$ was chosen to be $|\argmaxf{\beta}|$.
The variable $\ell$ is used to determine when $\beta$ should be incremented:
$\beta$ is incremented during an iteration if and only if $\ell$ is set to 0.

\subsubsection{Analysis of \kappabea for SM}
The approximation results of \kappabea for SM are now presented.
The statement of Theorem \ref{theorem:kappabea} easily generalizes to $\gamma$-weakly
submodular objectives $f$ in the analogous manner as \bea.

\begin{theorem}
  \label{theorem:kappabea}
  Suppose \kappabea is run with input $(f,\kappa,P,T,p,\epsilon)$
  where \definef is monotone submodular, $\kappa\in\{1,...,n\}$
  $P\in\{\kappa+1,...,n+1\}$,
  $T\geq\max\{2en\ln(1/\epsilon)/p, 8\ln(n)/p\}$,
  \definep, and \defineepsilon.
  Let \poolb be the pool of \kappabea at the end of iteration $T$.
  Then,
  $$\ex{\max_{X\in\mathcal{S}, |X|\leq \kappa}f(X)}
  \geq (1-\epsilon)(1-1/e-\epsilon)\max_{|X|\leq \kappa}f(X).$$
\end{theorem}
\begin{proof}
  The proof of Theorem \ref{theorem:kappabea} is any easy modification of the proof of
  Theorem \ref{theorem:bea} and therefore details are left to the reader.
  The key point is that Lemma \ref{lemma:bernoullibea} should be replaced with
  the following lemma.
  \begin{lemma}
    \label{lemma:bernoullikappabea}
    Consider a run of \kappabea as a series of Bernoulli trials $Y_1,...,Y_T$,
    where each iteration is a
    trial and a success is defined to be when $\ell$ is incremented.
    Then $Y_1,...,Y_T$ are independent, identically distributed Bernoulli trials where the
    probability of success is $p$.
  \end{lemma}
\end{proof}

%

%


%

\section{Experimental Evaluation}
\label{section:experiments}
\begin{figure*}[t!]
  \centering
  \subfigure[Gaussian, $f_{FL}$, \budget$=0.1n$] {
    \label{fig:fb-sub-early}
    \includegraphics[width=0.31\textwidth,height=0.17\textheight]{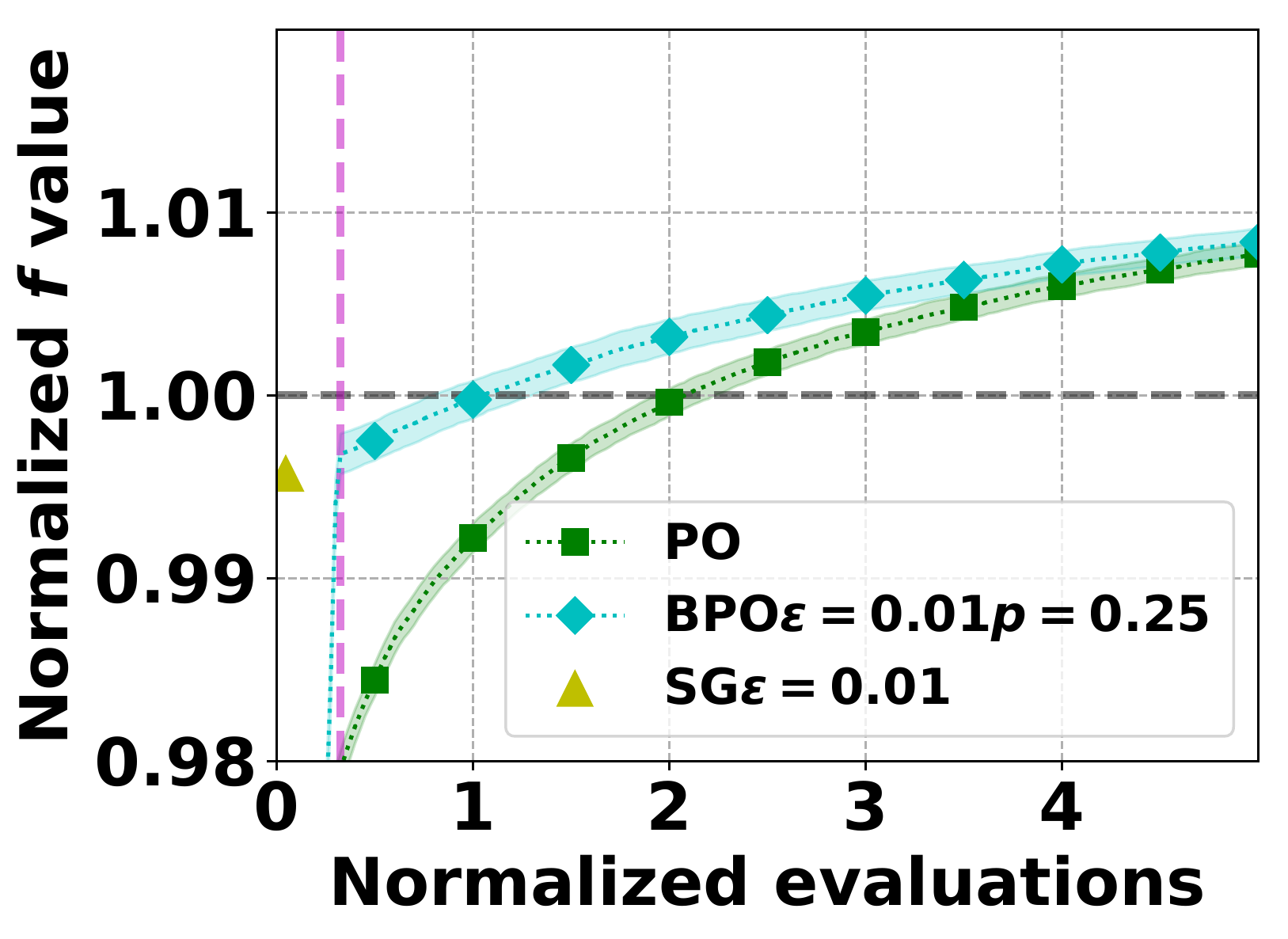}
  }
  \subfigure[CIFAR, $f_{FL}$, \budget$=0.1n$] {
    \label{fig:fb-sub-all}
    \includegraphics[width=0.31\textwidth,height=0.17\textheight]{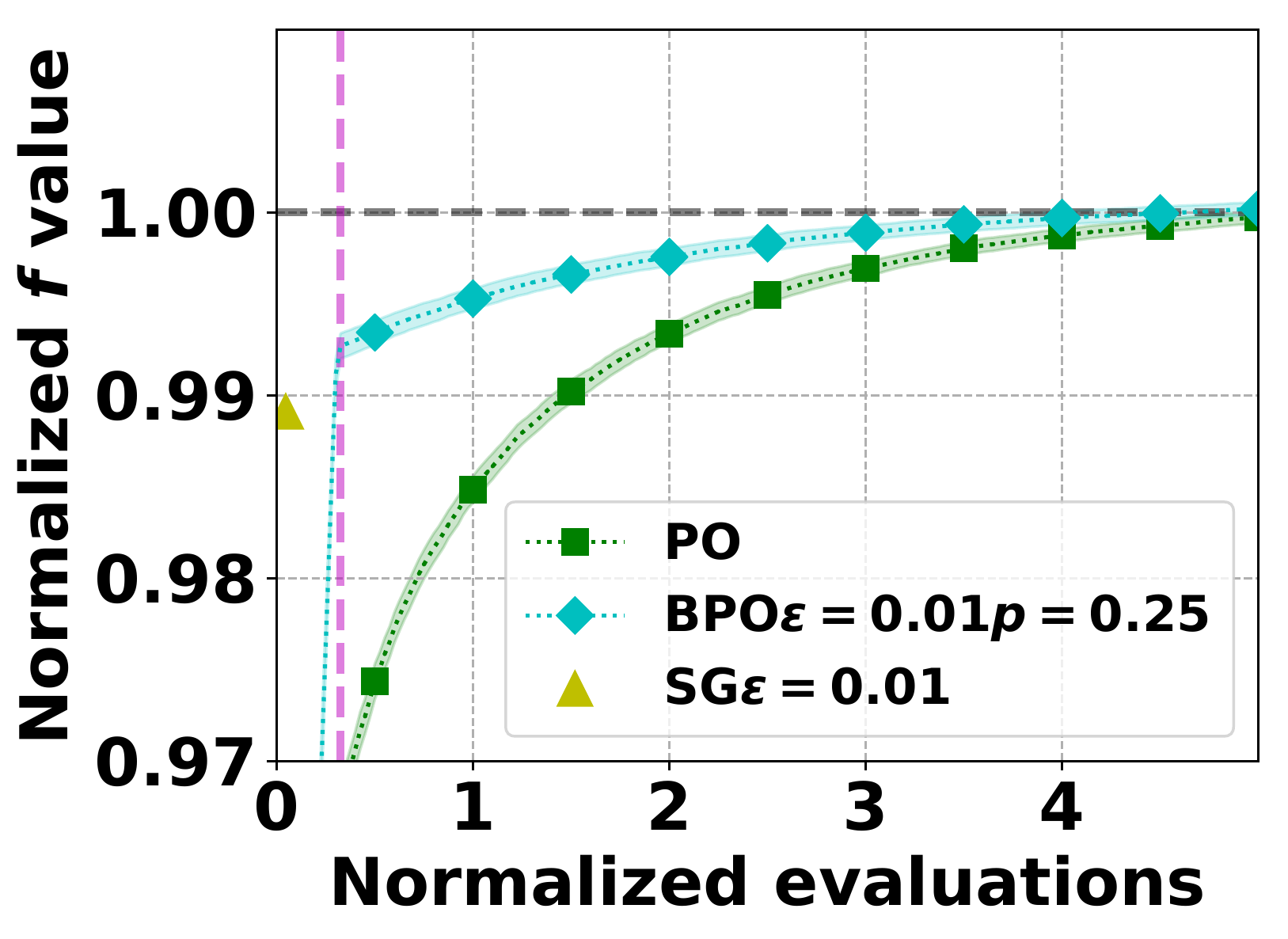}
  }
  \subfigure[Gaussian, $f_{FL}$, \budget$=0.1n$] {
    \label{fig:orkut}
    \includegraphics[width=0.31\textwidth,height=0.17\textheight]{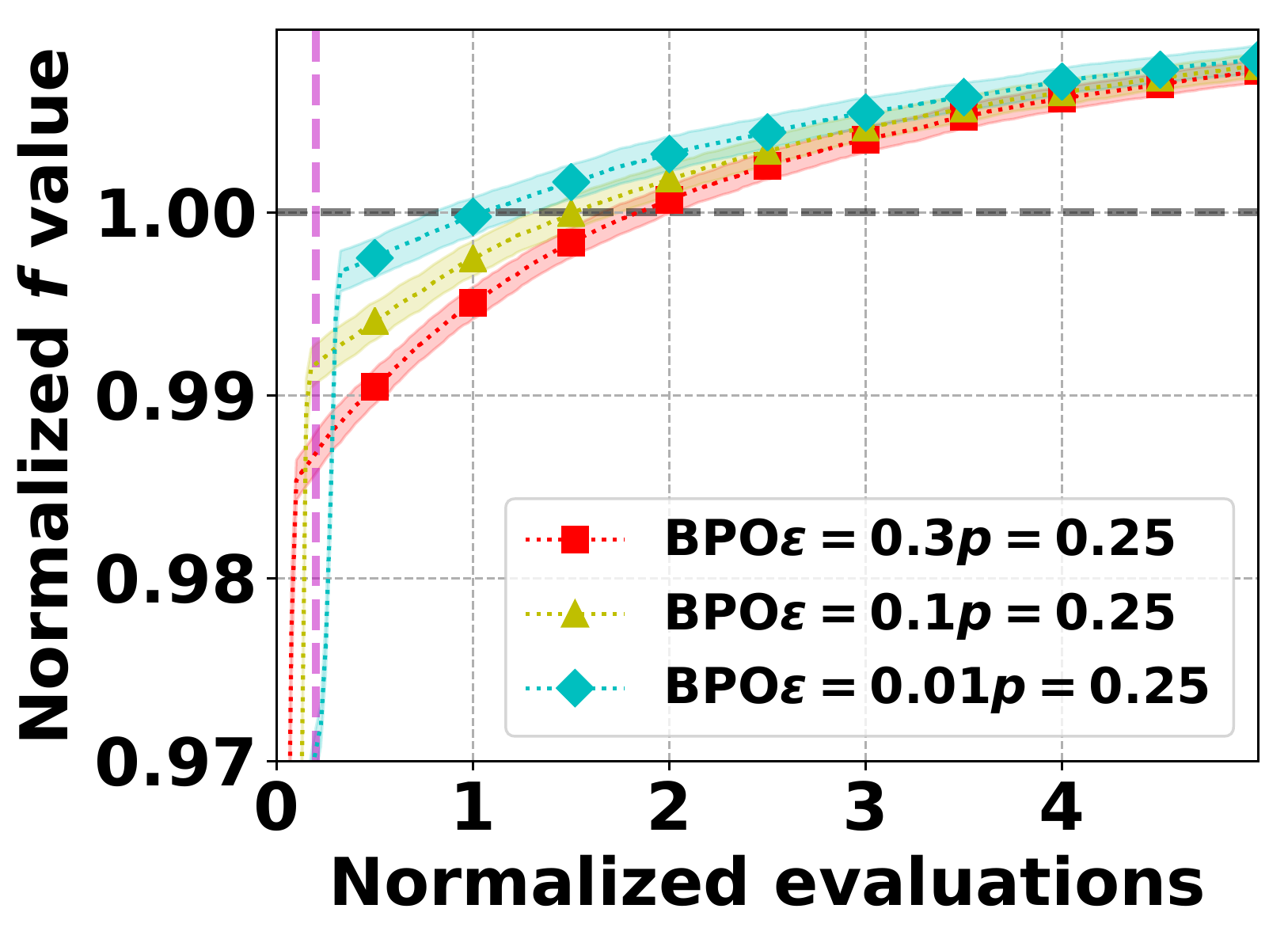}
  }
  \subfigure[Gaussian, $f_{FL}$, \budget$=0.1n$] {
    \label{fig:fb-nonsub-early}
    \includegraphics[width=0.31\textwidth,height=0.17\textheight]{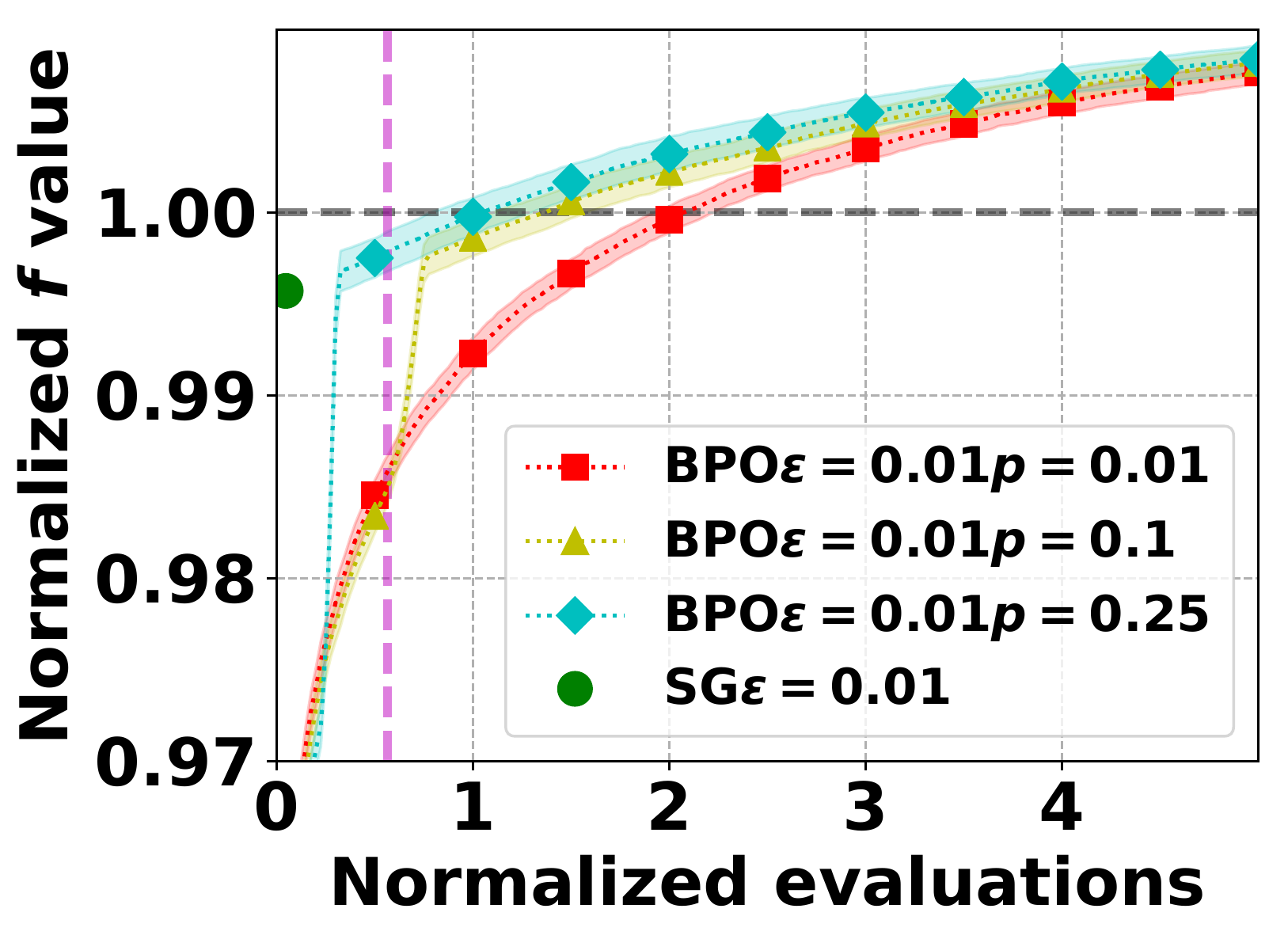}
  }
  \subfigure[Gaussian, $f_{DPP}$, \budget$=0.1n$] {
    \label{figure:gaussiandpp}
    \includegraphics[width=0.31\textwidth,height=0.17\textheight]{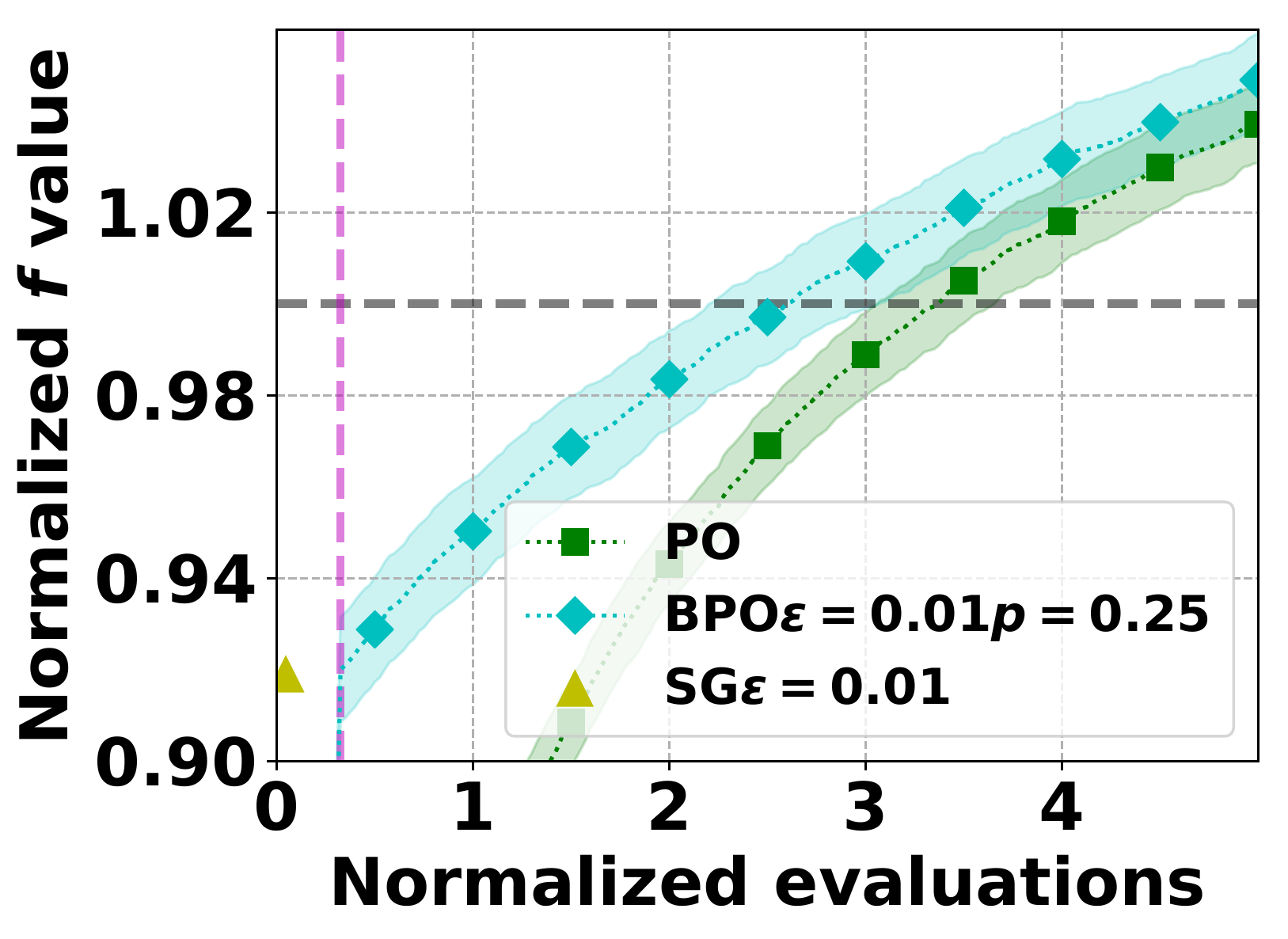}
  }
  \subfigure[CIFAR, $f_{DPP}$, \budget$=0.1n$] {
    \label{figure:cifardpp}
    \includegraphics[width=0.31\textwidth,height=0.17\textheight]{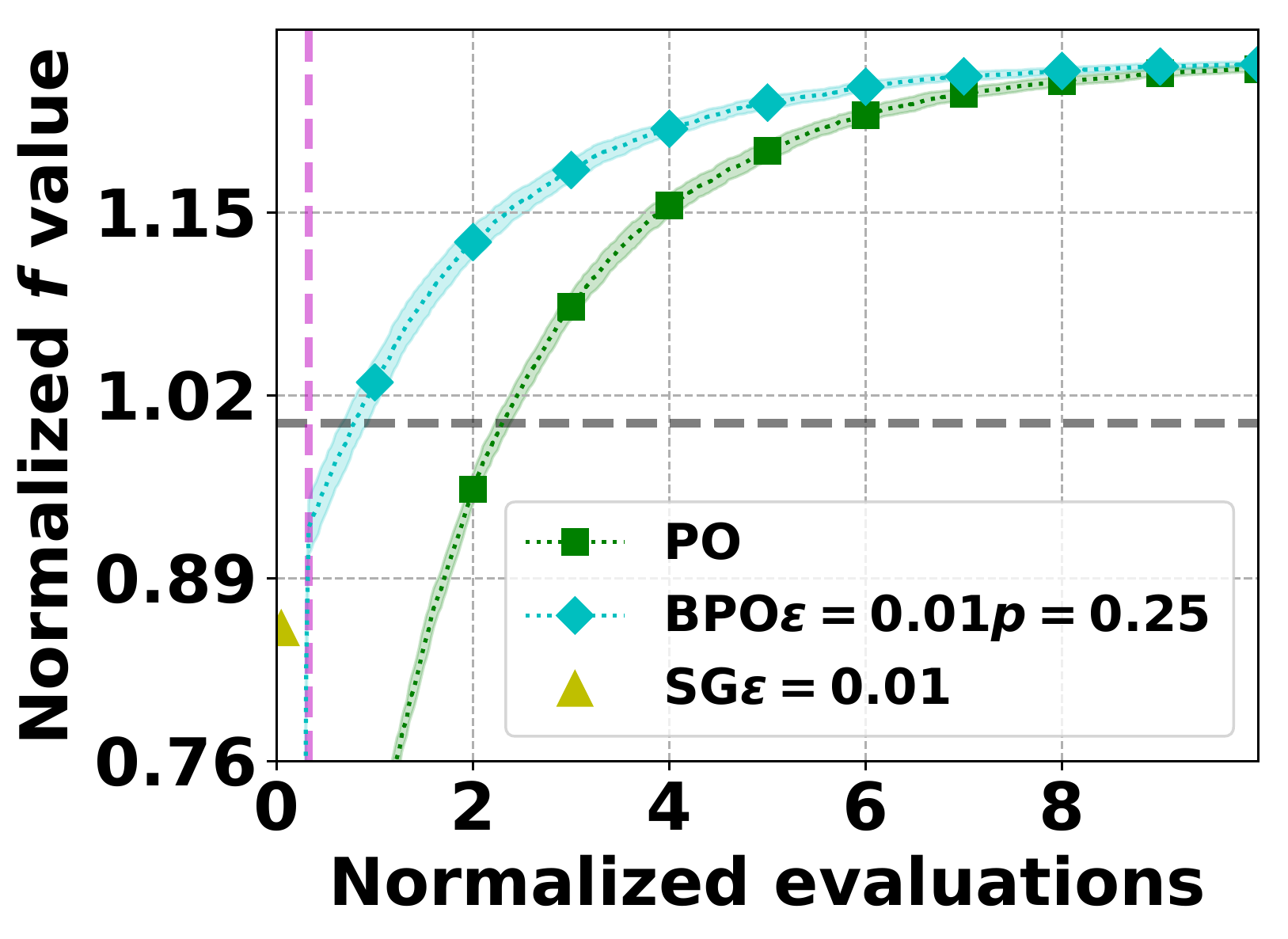}
  }
  \caption{In all plots the $y$-axis is normalized by the standard greedy value on the instance,
    and the $x$-axis is normalized by the number $kn$
    of evaluations required by the standard greedy algorithm.
    The dataset, objective, and value of $\kappa$ are indicated in the caption of each subfigure.}
    \label{fig:exp-main}
\end{figure*}

In this section, the algorithms \ea and \kappabea are evaluated on instances
of data summarization with submodular and non-submodular
objectives $f$.
In summary, the faster runtime for \ea
proven in Theorem \ref{theorem:ea} is demonstrated empirically.
Also,
the results demonstrate that \kappabea quickly finds solutions better than
the standard greedy algorithm, the stochastic greedy algorithm \sg \cite{mirzasoleiman2015lazier}, and \ea.

The algorithms evaluated in Section \ref{section:experiments} are:
\begin{itemize}
\item the standard greedy algorithm \cite{Nemhauser1978}
\item the stochastic greedy (\sg) algorithm of
  \mirza.
\item \ea: the variant of the algorithm of \friedt
  as detailed in Alg. \ref{algorithm:ea} and analyzed in
  Section \ref{section:ea}.
\item \kbea: the version of \bea that biases towards
  only one set in $\pool$, based on the input $\kappa$
  as discussed in Section \ref{section:kappabea}.
\end{itemize}

For both \ea and \kbea, the parameter $P=2\kappa$ is
used on all instances.

\subsection{Application}
In data summarization (DS), we have a set $U$ of data points and we wish to find
a subset of $U$ of cardinality $\kappa$ that best \textit{summarizes} the entire
dataset $U$. \definef takes $X\subseteq U$ to a measure of how effectively $X$
summarizes $U$.
For the ground set $U$, we use: (i) A set of 10 dimensional vectors
drawn from $\kappa$ gaussian distributions (Gaussian), and
(ii) a set of $32\times 32$ color images from the CIFAR-100
dataset \cite{krizhevsky2009learning} each represented by a 3072 dimensional
vector of pixels (CIFAR).
For the objective $f$, we use:
(i) The monotonic and submodular objective $k$-medoid objective \cite{kaufman2009finding}
($f_{MED}$), and (ii) the monotone weakly submodular
objective based on Determinantal Point Process (DPP) \cite{kulesza2012determinantal}
($f_{DPP}$).
A lower bound on the submodularity ratio has been proven for the latter
objective \cite{bian2017guarantees}.

\subsection{Results}
The experimental results are shown in Figure \ref{fig:exp-main}.
All results are the mean of $50$ repetitions
of each algorithm; shaded regions represent one standard deviation from the mean.
Objective and runtime are normalized by the objective value of
and number of queries made ($n\kappa$) by the standard greedy
algorithm. The value of the solution of
the standard greedy algorithm is plotted as a dotted gray horizontal line $y=1$.
The time where $\beta$ in \kbea reached \budget is plotted as a vertical magenta line.

The best solution value
obtained by each algorithm is shown as the rightmost
point in each plot. Both \kbea and \ea were eventually able to find
better solutions than the standard greedy algorithm (i.e., normalized value $> 1.0$), especially on the
non-submodular objective (Figures \ref{figure:gaussiandpp} and \ref{figure:cifardpp}).
Observe that \ea typically exceeds the stochastic greedy objective value
within $c\kappa n$, where $c \le 2$.
This behavior corroborates our theoretical analysis that \ea achieves
a good solution in expectation in $O(\kappa n)$ queries.
In addition, \kbea exceeds the SG value in $cn$ queries.
Finally, for \ea, recall that the theoretical anlaysis shows
that for any $\kappa < P$, the approximation ratio holds.

Because \ea and \kbea can be terminated at any time,
the running time may be compared by observing where any vertical line
intersects the curves for each algorithm. The running time of the
standard greedy corresponds to the line $x=1$ (not plotted).
\kbea reaches solution values closer to the standard greedy algorithm in significantly
faster time than \ea, as expected by its design.
The effect of varying the parameters $\varepsilon$ and $p$ on the behavior of \kbea
is shown in Figs. \ref{fig:orkut}, \ref{fig:fb-nonsub-early},
respectively: smaller $\varepsilon$ leads
to a higher initial increase but the initial increase is slower,
while smaller $p$ slows down the rate of the initial increase.

%


\section{Conclusions}
In this work, we have re-analyzed
the evolutionary algorithm \ea,
originally analyzed for submodular
maximization by \friedt, and showed
that it achieves nearly the optimal worst-case ratio
in expectation
on SM for any $\kappa < P$ in
$O(nP)$ queries.
In contrast, \friedt showed that the optimal worst-case ratio
is achieved in expected $O(nP^2)$ queries.
This improved rate of convergence is supported
by an empirical evaluation.

Further,
it has been shown that changing
the selection process in \ea
results in improved query complexity
to $O(n \log( P ))$ to obtain the same
approximation results. A variant of this algorithm \kbea
is shown empirically to have a much faster
initial rate of convergence to a good solution
than \ea, without sacrificing the long-term
behavior of the \ea algorithm.


\clearpage
\bibliographystyle{style/named}
\bibliography{Bibliography/paper,Bibliography/ijcai2019,Bibliography/old,Bibliography/mend}{}

\begin{thebibliography}{}

\bibitem[\protect\citeauthoryear{Badanidiyuru \bgroup \em et al.\egroup
  }{2014}]{Badanidiyuru2014}
Ashwinkumar Badanidiyuru, Baharan Mirzasoleiman, Amin Karbasi, and Andreas
  Krause.
\newblock {Streaming submodular maximization: massive data summarization on the
  fly}.
\newblock In {\em ACM SIGKDD International Conference on Knowledge Discovery
  and Data Mining (KDD)}, 2014.

\bibitem[\protect\citeauthoryear{Bian \bgroup \em et al.\egroup
  }{2017}]{bian2017guarantees}
Andrew~An Bian, Joachim~M Buhmann, Andreas Krause, and Sebastian Tschiatschek.
\newblock Guarantees for greedy maximization of non-submodular functions with
  applications.
\newblock In {\em International Conference on Machine Learning (ICML)}, 2017.

\bibitem[\protect\citeauthoryear{Bian \bgroup \em et al.\egroup
  }{2020}]{bian2020efficient}
Chao Bian, Chao Feng, Chao Qian, and Yang Yu.
\newblock An efficient evolutionary algorithm for subset selection with general
  cost constraints.
\newblock In {\em AAAI Conference on Artificial Intelligence (AAAI)}, 2020.

\bibitem[\protect\citeauthoryear{Crawford}{2019}]{crawford19}
Victoria~G. Crawford.
\newblock An efficient evolutionary algorithm for minimum cost submodular
  cover.
\newblock In {\em International Joint Conference on Artificial Intelligence
  (IJCAI)}, 2019.

\bibitem[\protect\citeauthoryear{Das and Kempe}{2011}]{Das2011}
Abhimanyu Das and David Kempe.
\newblock {Submodular meets Spectral: Greedy Algorithms for Subset Selection,
  Sparse Approximation and Dictionary Selection}.
\newblock {\em Proceedings of the 28th International Conference on Machine
  Learning (ICML)}, 2011.

\bibitem[\protect\citeauthoryear{Friedrich and Neumann}{2014}]{Friedrich2014}
Tobias Friedrich and Frank Neumann.
\newblock {Maximizing submodular functions under matroid constraints by
  evolutionary algorithms}.
\newblock In {\em International Conference on Parallel Problem Solving from
  Nature (PPSN)}, 2014.

\bibitem[\protect\citeauthoryear{Friedrich \bgroup \em et al.\egroup
  }{2010}]{friedrich2010approximating}
Tobias Friedrich, Jun He, Nils Hebbinghaus, Frank Neumann, and Carsten Witt.
\newblock Approximating covering problems by randomized search heuristics using
  multi-objective models.
\newblock {\em Evolutionary Computation}, 18(4):617--633, 2010.

\bibitem[\protect\citeauthoryear{Kaufman and
  Rousseeuw}{2009}]{kaufman2009finding}
Leonard Kaufman and Peter~J Rousseeuw.
\newblock {\em Finding groups in data: an introduction to cluster analysis},
  volume 344.
\newblock John Wiley \& Sons, 2009.

\bibitem[\protect\citeauthoryear{Kempe \bgroup \em et al.\egroup
  }{2003}]{kempe2003maximizing}
David Kempe, Jon Kleinberg, and {\'E}va Tardos.
\newblock Maximizing the spread of influence through a social network.
\newblock In {\em Proceedings of the ninth ACM SIGKDD international conference
  on Knowledge discovery and data mining}, pages 137--146. ACM, 2003.

\bibitem[\protect\citeauthoryear{Krizhevsky \bgroup \em et al.\egroup
  }{2009}]{krizhevsky2009learning}
Alex Krizhevsky, Geoffrey Hinton, et~al.
\newblock {\em Learning multiple layers of features from tiny images}.
\newblock Technical Report, University of Toronto, 2009.

\bibitem[\protect\citeauthoryear{Kulesza \bgroup \em et al.\egroup
  }{2012}]{kulesza2012determinantal}
Alex Kulesza, Ben Taskar, et~al.
\newblock Determinantal point processes for machine learning.
\newblock {\em Foundations and Trends{\textregistered} in Machine Learning},
  5(2--3):123--286, 2012.

\bibitem[\protect\citeauthoryear{Laumanns \bgroup \em et al.\egroup
  }{2002}]{laumanns2002running}
Marco Laumanns, Lothar Thiele, Eckart Zitzler, Emo Welzl, and Kalyanmoy Deb.
\newblock Running time analysis of multi-objective evolutionary algorithms on a
  simple discrete optimization problem.
\newblock In {\em International Conference on Parallel Problem Solving from
  Nature (PPSN)}, 2002.

\bibitem[\protect\citeauthoryear{Mirzasoleiman \bgroup \em et al.\egroup
  }{2013}]{mirzasoleiman2013distributed}
Baharan Mirzasoleiman, Amin Karbasi, Rik Sarkar, and Andreas Krause.
\newblock Distributed submodular maximization: Identifying representative
  elements in massive data.
\newblock In {\em Advances in Neural Information Processing Systems (NeurIPS)},
  2013.

\bibitem[\protect\citeauthoryear{Mirzasoleiman \bgroup \em et al.\egroup
  }{2015}]{mirzasoleiman2015lazier}
Baharan Mirzasoleiman, Ashwinkumar Badanidiyuru, Amin Karbasi, Jan Vondrak, and
  Andreas Krause.
\newblock {Lazier Than Lazy Greedy}.
\newblock In {\em AAAI Conference on Artificial Intelligence (AAAI)}, 2015.

\bibitem[\protect\citeauthoryear{Nemhauser and Wolsey}{1978}]{Nemhauser1978}
G~L Nemhauser and L~A Wolsey.
\newblock {Best Algorithms for Approximating the Maximum of a Submodular Set
  Function}.
\newblock {\em Mathematics of Operations Research}, 3(3):177--188, 1978.

\bibitem[\protect\citeauthoryear{Neumann and
  Wegener}{2007}]{neumann2007randomized}
Frank Neumann and Ingo Wegener.
\newblock Randomized local search, evolutionary algorithms, and the minimum
  spanning tree problem.
\newblock {\em Theoretical Computer Science}, 378(1):32--40, 2007.

\bibitem[\protect\citeauthoryear{Qian \bgroup \em et al.\egroup
  }{2015a}]{qian2015constrained}
Chao Qian, Yang Yu, and Zhi-Hua Zhou.
\newblock On constrained boolean pareto optimization.
\newblock In {\em International Joint Conference on Artificial Intelligence
  (IJCAI)}, 2015.

\bibitem[\protect\citeauthoryear{Qian \bgroup \em et al.\egroup
  }{2015b}]{qian2015subset}
Chao Qian, Yang Yu, and Zhi-Hua Zhou.
\newblock Subset selection by pareto optimization.
\newblock In {\em Advances in Neural Information Processing Systems (NeurIPS)},
  2015.

\bibitem[\protect\citeauthoryear{Qian \bgroup \em et al.\egroup
  }{2017}]{Qian2017}
Chao Qian, Jing-cheng Shi, Yang Yu, Ke~Tang, and Zhi-hua Zhou.
\newblock {Subset Selection under Noise}.
\newblock In {\em Advances in Neural Information Processing Systems (NeurIPS)},
  2017.

\bibitem[\protect\citeauthoryear{Roostapour \bgroup \em et al.\egroup
  }{2019}]{roostapour2019pareto}
Vahid Roostapour, Aneta Neumann, Frank Neumann, and Tobias Friedrich.
\newblock Pareto optimization for subset selection with dynamic cost
  constraints.
\newblock In {\em AAAI Conference on Artificial Intelligence (AAAI)}, 2019.

\bibitem[\protect\citeauthoryear{Soma and Yoshida}{2016}]{Soma2015a}
Tasuku Soma and Yuichi Yoshida.
\newblock Maximizing monotone submodular functions over the integer lattice.
\newblock In {\em International Conference on Integer Programming and
  Combinatorial Optimization (IPCO)}, 2016.

\end{thebibliography}
\onecolumn
\section{Appendix: Algorithms and Theoretical Results}
\label{appendix:theoretical}
In this section, additional results from Section \ref{section:theoretical} are included.
In particular:
Pseudocode missing from the description of the algorithm \ea is given
in Section \ref{appendix:eaoverview};
Lemmas used in Section \ref{section:ea} and their corresponding proofs are given
in Section \ref{appendix:ealemmas};
Lemmas used in Section \ref{section:bea} and their corresponding proofs are given
in Section \ref{appendix:bealemmas};
Finally, pseudocode for \kappabea is given in Section \ref{appendix:kappabea}

\subsection{\eal (\ea)}
\label{appendix:ea}
\subsubsection{Additional Pseudocode}
\label{appendix:eaoverview}


\begin{algorithm}
   \caption{\mutate$(B)$}
   \label{algorithm:mutate}
   \begin{algorithmic}[1]
     \STATE {\bfseries Input:} $B \subseteq U$.
     \STATE {\bfseries Output:} $B' \subseteq U$.
     \STATE $B'\gets B$
     \FOR{$x \in U$}
     \IF{$x \in B$}
     \STATE With probability $1/n$, $B' \gets B' \setminus \{x\}$.
     \ELSE
     \STATE With probability $1/n$, $B' \gets B' \cup \{x\}$.
     \ENDIF
     \ENDFOR
     \STATE \textbf{return} $B'$
\end{algorithmic}
\end{algorithm}

\subsubsection{Lemmas used for the Proof of Theorem \ref{theorem:ea}}
\label{appendix:ealemmas}
\begin{paragraph}{Lemma \ref{lemma:expectation}.}
  {\it
  At the end of every iteration of \ea
  $$\ex{\betaf} \geq \left(1-\left(1-\frac{1}{|A^*|}\right)^{\omega'}\right)f(A^*)$$
  where $\omega'=\min\{\omega,P-1\}$.}
  \end{paragraph}


\begin{proof}
%
  At the end of any iteration $i$ of \ea, define (i) $\omega_i$ to be the value of
  $\omega$,
  and (ii) $X_i=\betasetline$.
  Define $\omega_0$ and $X_0$ refer to the values at the start of the
  first iteration.
  In order to prove Lemma \ref{lemma:expectation}, it must be
  shown that at the end of any iteration $i\in\{0,1,...,T\}$,
  \begin{align}
    \label{eqn:induction}
    \ex{f(X_i)} \geq \left(1-\left(1-\frac{1}{|A^*|}\right)^{\min\{\omega_i, P-1\}}\right)f(A^*).
  \end{align}
  Equation \ref{eqn:induction} will be proven by induction on
  iteration $i$.
  Equation \ref{eqn:induction} is clearly true for $i=0$,
  since on all runs of \bea $\omega_0=0$ and $X_0=\emptyset$.
  Now suppose that Equation \ref{eqn:induction} is true
  for iteration $t-1\in\{0,...,T-1\}$; it will be shown that it is then true for iteration $t$.

  \noindent Define $E$ to be the event that during iteration $t$ of the for loop
  of \ea, $\omega$ is incremented by 1.
  The following claim establishes that
  expected value of $f(X_{t-1})$ does not depend
  on $E$.
  \begin{claim}
    \label{claim:independenceE}
    $\ex{f(X_{t-1})| E}=\ex{f(X_{t-1})}.$
  \end{claim}
  \begin{proof}
      At the beginning of iteration $t$ the probability that
      $\omega$ will be incremented is
      \begin{align*}
        \frac{1}{P}\sum_{x\in A^*}\left(1-\frac{1}{n}\right)^{n-1}\frac{1}{n}.
      \end{align*}
      This does not depend on the value of $f(X_{t-1})$, so
      for all $\alpha\in\mathbb{R}_{\geq 0}$, $P(E|f(X_{t-1})=\alpha)=P(E)$.
      Then Bayes' theorem gives that
      $P(f(X_{t-1})=\alpha|E) = P(f(X_{t-1})=\alpha)$,
      which implies Claim \ref{claim:independenceE}.
  \end{proof}

\noindent  Inequality \ref{eqn:induction} is proven by breaking up into the two cases
  $E$ and $\neg E$, and then applying the law of total probability.
  If $E$ did not occur,
  then it follows from Lemma \ref{lemma:replace} that
  \begin{align}
    \ex{f(X_t)|\neg E} &\geq \ex{f(X_{t-1})| \neg E } \nonumber\\
    &= \ex{f(X_{t-1}) } \nonumber \\
    &\ge \left(1-\left(1-\frac{1}{|A^*|}\right)^{\min\{\omega_t, P-1\}}\right)f(A^*), \label{eqn:327899}
  \end{align}
  where the last inequality follows from the inductive assumption
  and the fact that $\omega_t = \omega_{t-1}$ when conditioning
  on $\neg E$.

  The proof will proceed by
  considering arbitrary but fixed runs of \ea.
  Consider runs of \ea where $E$ did occur.
  If $E$ occurs, then
  during iteration $t$, $X_{t-1}$ is selected
  on Line \ref{line:selectea} and
  then \mutate results in the membership of a single $a^*\in A^*$ being flipped
  in $X_{t-1}$, i.e. it is either the case that \mutate returns $X_{t-1}\cup\{a^*\}$
  for $a^*\in A^*\setminus X_{t-1}$
  or $X_{t-1}\setminus\{a^*\}$ for $a^*\in A^*\cap X_{t-1}$.
  \begin{claim}
    \label{claim:atleast}
    The following holds:
    $\ex{f(X_t)|E} \geq \ex{f(X_{t-1}\cup\{a^*\})|E}.$
  \end{claim}
  \begin{proof}
    If $E$ occurs, then \mutate returns either $X_{t-1}\cup\{a^*\}$ for $a^*\in A^*\setminus X_{t-1}$
    or $X_{t-1}\setminus\{a^*\}$ for $a^*\in A^*\cap X_{t-1}$. Let $E_1$ be the former
    event and let $E_2$ be the latter event.

    Consider a particular run of \ea where $E$ occurs.
    Then $|X_{t-1}\cup\{a^*\}|\leq\omega_{t-1}+1=\omega_t$.
    If $E_1$ occurs on this run, then $X_{t-1}\cup\{a^*\}$ was returned by \mutate
    on iteration $t$
    and therefore $f(X_t)\geq f(X_{t-1}\cup\{a^*\})$ by Lemma \ref{lemma:replace}.
    If $E_2$ occurs on this run, then $X_{t-1}\cup\{a^*\}=X_{t-1}$. $X_{t-1}$ was
    returned by \mutate on some run before $t$ since $X_{t-1}\in\pool$ at the end
    of iteration $t-1$. Therefore $f(X_t)\geq f(X_{t-1}\cup\{a^*\})$ by Lemma
    \ref{lemma:replace}.
    Then for any run of \ea where $E$ occurs $f(X_t)\geq f(X_{t-1}\cup\{a^*\})$
    and so Claim \ref{claim:atleast} follows.
%
%
  \end{proof}

  Next, it holds that
  \begin{align}
    \ex{f(X_t)|E} &\overset{(a)}{\geq} \ex{f(X_{t-1}\cup\{a^*\})|E} \nonumber \\
    &= \ex{f(X_{t-1})|E} + \ex{f(X_{t-1}\cup\{a^*\})-f(X_{t-1})|E}\nonumber \\
    &\overset{(b)}{\geq} \ex{f(X_{t-1})|E} + \frac{1}{|A^*|}(f(A^*)-\ex{f(X_{t-1})}|E)\nonumber \\
    &\overset{(c)}{=} \ex{f(X_{t-1})} + \frac{1}{|A^*|}(f(A^*)-\ex{f(X_{t-1})})\nonumber \\
    &\overset{(d)}{\geq} \left(1-\left(1-\frac{1}{|A^*|}\right)^{\omega_{t-1}+1}\right)f(A^*)\nonumber \\
    &\overset{(e)}{=} \left(1-\left(1-\frac{1}{|A^*|}\right)^{\omega_t}\right)f(A^*) \label{eqn:xssda3}
  \end{align}
  where (a) follows from Claim \ref{claim:atleast};
  (b) follows from Lemma \ref{lemma:gain};
  (c) is by Claim \ref{claim:independenceE};
  (d) is the inductive assumption;
  (e) is because by definition of $E$, $\omega_t=\omega_{t-1}+1$.
  Finally, the inductive step follows by Equations \ref{eqn:327899}, \ref{eqn:xssda3}
  and the law of total probability. Therefore
  Lemma \ref{lemma:expectation} is proven.
\end{proof}


\begin{paragraph}{Lemma \ref{lemma:replace}}
  \it{
    Let $Y\subseteq U$ and $|Y|\leq a < P$. If $Y$ is returned by \mutate during
    iteration $i$ of \ea, then at the end of any iteration $j\geq i$ it holds that
    $\maxf{a} \geq f(Y).$
  }
\end{paragraph}
\begin{proof}
  First, notice that once $Y$ has been returned by \mutate, there exists
  $X\in\pool$ such that $Y\preceq X$ from that point in \ea on.
  Consider the end of any iteration $j\geq i$. Let $Y'$ be $Y'\in\pool$ such that $Y\preceq Y'$
  at the end of iteration $j$.
  Then $f(Y)\leq f(Y')\leq \maxf{a}$ since $|Y'|\leq |Y|\leq a$.
\end{proof}

\begin{paragraph}{Lemma \ref{lemma:bernoulli}.}
  Consider a run of \ea as a series of Bernoulli trials $Y_1,...,Y_T$,
  where each iteration is a
  trial and a success is defined to be when $\omega$ is incremented.
  $Y_1,...,Y_T$ are independent, identically distributed Bernoulli trials where the
  probability of success is
  \begin{align*}
    \frac{1}{P}\sum_{x\in A^*}\left(1-\frac{1}{n}\right)^{n-1}\frac{1}{n} \geq \frac{|A^*|}{enP}.
  \end{align*}
\end{paragraph}
\begin{proof}
  In order for $\omega$ to be incremented on an iteration $t$ (i.e. the trial results in a success),
  the set $X_{t-1}$ must be selected on Line \ref{line:selectea}; this occurs with probability $1/P$.
  In addition, if any set is selected, \mutate
  results in flipping of the membership of a single $a^*\in A^*$ and nothing else,
  with probability
  $\sum_{x\in A^*}\left(1-\frac{1}{n}\right)^{n-1}\frac{1}{n}.$  Therefore,
  the probability of success is $\frac{1}{P} \sum_{x\in A^*}\left(1-\frac{1}{n}\right)^{n-1}\frac{1}{n}$.
  This probability is independent of the iteration $t$, and therefore it
  follows that $Y_1,...,Y_T$ are independent and identically distributed.
\end{proof}

\begin{paragraph}{Lemma \ref{lemma:probability}.}
  Consider a run of \ea as a series of Bernoulli trials $Y_1,...,Y_T$,
  where each iteration is a
  trial and a success is defined to be when $\omega$ is incremented. Then
  \begin{align*}
    P\left(\sum_{i=1}^TY_i < \kappa\right) \leq \frac{1}{n}.
  \end{align*}
\end{paragraph}
\begin{proof}
By Lemma \ref{lemma:bernoulli}, Chernoff's bound
(Lemma \ref{lemma:chernoff}) may be applied to $Y_1,...,Y_T$.
Let $\rho$ be the probability of success for each $Y_i$
(the value of $\rho$ is given in Lemma \ref{lemma:bernoulli}).
Then $\ex{\sum_{i=1}^TY_i}=T\rho$.
Therefore
\begin{align*}
  P\left(\sum_{i=1}^TY_i < \kappa\right)
  &\overset{(a)}{\leq} P\left(\sum_{i=1}^TY_i < T\rho/2\right) \\
  &\overset{(b)}{\leq} e^{-T\rho/8} \\
  &\overset{(c)}{\leq} \epsilon
\end{align*}
where (a) is because $T\geq\eaTa$ combined with the lower bound on $\rho$ given in
Lemma \ref{lemma:bernoulli}; (b) is applying Lemma \ref{lemma:chernoff} with $\eta=1/2$;
and (c) is because $T\geq\eaTb$ combined with the lower bound on $\rho$ given in
Lemma \ref{lemma:bernoulli}.
\end{proof}

\begin{lemma}
  \label{lemma:gain}
  Let $B,X\subseteq U$, $X\neq\emptyset$.
  Suppose that $B$ is input to \mutate, and consider the probability space of
  all possible outputs of \mutate. Let $E$ be the event that \mutate returns
  $B\cup\{x\}$ for some $x\in X\setminus B$ or $B\setminus\{x\}$ for some $x\in X\cap B$. Then
  \begin{align*}
    \ex{\delt{B}{x}|E} \geq \frac{1}{|X|}(f(X)-f(B)).
  \end{align*}
\end{lemma}
\begin{proof}
  It is clear from the procedure \mutate that the event $E$ is equivalent to
  uniformly randomly choosing $x\in X$ and then flipping its membership in $B$.
  Then it is the case that
  \begin{align*}
    \ex{\delt{B}{x}|E} &= \frac{1}{|X|}\sum_{y\in X}\delt{B}{y} \\
                       &\overset{a}{\geq} \frac{1}{|X|}(f(X)-f(B))
  \end{align*}
  where (a) follows from the monotonicity and submodularity of $f$.
\end{proof}

\begin{lemma}[Chernoff bound]
  \label{lemma:chernoff}
  Suppose $Y_1,\ldots,Y_T$ are independent random variables taking values in
  $\{0,1\}$. Let $Y$ denote their sum and let $\mu = \ex{Y}$ denote
  the sum's expected value. Then for any $\eta > 0$
  \begin{align*}
    P( Y \le (1 - \eta) \mu ) \le e^{-\eta^2\mu/2}.
  \end{align*}
\end{lemma}


\subsection{\beal (\bea)}
\subsubsection{Lemmas used for the Proof of Theorem \ref{theorem:bea}}
\label{appendix:bealemmas}
\begin{paragraph}{Lemma \ref{lemma:expectationbea}}
  \it{
  At the end of every iteration of \bea
  $$\ex{\maxfb{\omega}} \geq \left(1-\left(1-\frac{1-\epsilon}{|A^*|}\right)^{\omega'}\right)f(A^*)$$
  where $\omega'=\min\{\omega,P-1\}$.
  }
\end{paragraph}
\begin{proof}
  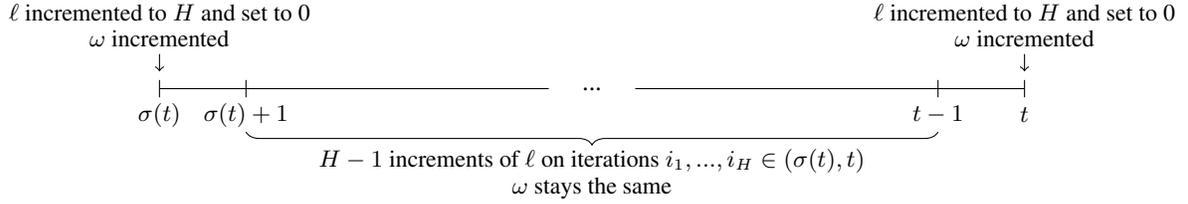
\begin{figure*}[b]
\centering
\begin{tikzpicture}[scale=1.15]
\draw (0,0) -- (4.5,0);
\node at (5,0) {...};
\draw (5.5,0) -- (10,0);

\draw (0,-0.1) -- (0,0.1);
\draw (1,-0.1) -- (1,0.1);
\draw (9,-0.1) -- (9,0.1);
\draw (10,-0.1) -- (10,0.1);

\node at (0,-0.3) {\small$\sigma(t)$};
\node at (1,-0.3) {\small$\sigma(t)+1$};
\node at (9,-0.3) {\small$t-1$};
\node at (10,-0.3) {\small$t$};

\draw[->] (0,0.4) -- (0,0.2);
\draw[->] (10,0.4) -- (10,0.2);

\node at (0,0.7) {\small\begin{tabular}{c}$\ell$ incremented to $H$ and set to 0 \\ $\omega$ incremented\end{tabular}};
\node at (10,0.7) {\small\begin{tabular}{c}$\ell$ incremented to $H$ and set to 0 \\ $\omega$ incremented\end{tabular}};

\draw[decoration={brace,amplitude=5pt},decorate] (9,-0.5) -- (1,-0.5);

\node at (5.0,-1.0) {\small\begin{tabular}{c}$H-1$ increments of $\ell$ on iterations $i_1,...,i_H\in(\sigma(t),t)$ \\ $\omega$ stays the same \end{tabular}};

\end{tikzpicture}
\caption{An illustration of the event $E$ in the proof of Lemma \ref{lemma:expectationbea}.
The line describes events on iterations of \bea starting at $\sigma(t)$ and increasing to $t$.
$\sigma(t)$ is defined to be the iteration where $\omega$ was set to $\omega_{t-1}$.
On iteration $\sigma(t)$, $\ell$ is set to 0 (which prompts the increment of $\omega$).
Then $\ell$ is incremented from 0 to $H$ on some subset of iterations in ($\sigma(t)+1,t$]
including $t$.}

\label{fig:eventE}
\end{figure*}

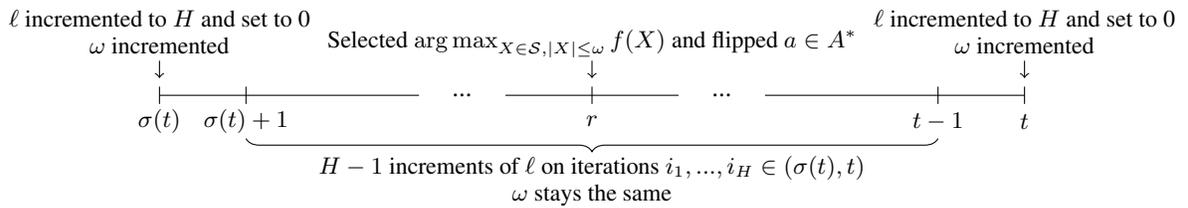
\begin{figure*}[b]
\centering
\begin{tikzpicture}[scale=1.15]
\draw (0,0) -- (3,0);
\node at (3.5,0) {...};
\draw (4,0) -- (6,0);
\node at (6.5,0) {...};
\draw (7,0) -- (10,0);

\draw (0,-0.1) -- (0,0.1);
\draw (1,-0.1) -- (1,0.1);
\draw (5,-0.1) -- (5,0.1);
\draw (9,-0.1) -- (9,0.1);
\draw (10,-0.1) -- (10,0.1);

\node at (0,-0.3) {\small$\sigma(t)$};
\node at (1,-0.3) {\small$\sigma(t)+1$};
\node at (5,-0.3) {\small$r$};
\node at (9,-0.3) {\small$t-1$};
\node at (10,-0.3) {\small$t$};

\draw[->] (0,0.4) -- (0,0.2);
\draw[->] (5,0.4) -- (5,0.2);
\draw[->] (10,0.4) -- (10,0.2);

\node at (0,0.7) {\small\begin{tabular}{c}$\ell$ incremented to $H$ and set to 0 \\ $\omega$ incremented\end{tabular}};
\node at (5,0.6) {\small Selected $\argmaxfb{\omega}$ and flipped $a\in A^*$};
\node at (10,0.7) {\small\begin{tabular}{c}$\ell$ incremented to $H$ and set to 0 \\ $\omega$ incremented\end{tabular}};

\draw[decoration={brace,amplitude=5pt},decorate] (9,-0.5) -- (1,-0.5);

\node at (5.0,-1.0) {\small\begin{tabular}{c}$H-1$ increments of $\ell$ on iterations $i_1,...,i_H\in(\sigma(t),t)$ \\ $\omega$ stays the same \end{tabular}};

\end{tikzpicture}
\caption{An illustration of the event $F$ in the proof of Lemma \ref{lemma:expectationbea}.
Event $F$ is a sub-event of the event $E$, which is depicted in Figure \ref{fig:eventE}.
In addition to the events described in $E$,
on one of the iterations $r$ in ($\sigma(t)+1,t$], the set $X_{r-1}$ is selected and then
\mutate results in the flipping of an element of $A^*$ and nothing else.}

\label{fig:eventF}
\end{figure*}

  For notational simplicity we define $H=H_q$, and $\ell=\ell_q$.
  $\omega_i$ and $X_i$ are defined analogously as in the proof of
  Lemma \ref{lemma:expectation}:
  At the end of any iteration $i$ of \bea, define (i) $\omega_i$ to be the value of
  $\omega$,
  and (ii) $X_i=\argmaxf{\omega}$.
  $\omega_0$ and $X_0$ refer to the values at the start of the
  first iteration.
  We see in Claim \ref{claim:improve} that the objective values of the sequence
  $X_1, X_2,....$ are non-decreasing.
  \begin{claim}
    \label{claim:improve}
    For any run of \bea, if $i\leq j$ then $f(X_i)\leq f(X_j)$.
  \end{claim}
  \begin{proof}
    $\omega$ only increases throughout \bea, therefore $\omega_i\leq\omega_j$.
    Then if \poolb is the pool at the end of iteration $j$
    \begin{align*}
      f(X_j) &= \maxf{\omega_j} \\
      &\geq \maxf{\omega_i} \\
      &\overset{a}{\geq} f(X_i)
    \end{align*}
    where (a) follows from Lemma \ref{lemma:replace}.
  \end{proof}

  In order to prove Lemma \ref{lemma:expectationbea} it must be
  shown that for any iteration $i\in\{0,1,...,T\}$,
  \begin{align}
    \label{eqn:inductionbea}
    \ex{f(X_i)} \geq \left(1-\left(1-\frac{1-\epsilon}{|A^*|}\right)^{\min\{\omega_i, P-1\}}\right)f(A^*)
  \end{align}
  which will be proven using induction on $i$.
  Equation \ref{eqn:inductionbea} is clearly true for $i=0$,
  since on all runs of \bea $\omega_0=0$ and $X_0=\emptyset$.
  Now suppose that Equation \ref{eqn:inductionbea} is true
  for every iteration $i < t\in\{1,...,T\}$;
  It will now be shown that Equation \ref{eqn:inductionbea} is true for iteration $t$.

  Define the random variable $\sigma(t)=\min\{i\in\{0,...,T-1\}:\omega_i=\omega_{t-1}\}$.
  In other words $\sigma(t)$ is the iteration during which $\omega$ was incremented to equal
  $\omega_{t-1}$, which is the value of $\omega$ at the beginning of iteration $t$.
  Then $\omega$ is not incremented on any iteration in $(\sigma(t),t)$.
  We will be using the inductive assumption on iteration $\sigma(t)$ in order to prove
  Equation \ref{eqn:inductionbea}.
  Define the events $E$ and $F$ as follows:
  \begin{itemize}[noitemsep]
    \item[(i)] Event $E$ is that during iteration $t$ of \bea, $\omega$ is incremented.
    Observe that if $E$ occurred, then during iteration $t$ of \bea $\ell$ was incremented to
    equal $H$. Since $\ell$ is set to 0 during iteration $\sigma(t)$,
    this implies that that on $H$ distinct
    iterations of \bea $i_1,...,i_H\in(\sigma(t),t]$ $\ell$ was incremented.
    In other words, on each of these iterations
    $j\in\{i_1,...,i_H\}$ the set $X_{j-1}$ is chosen for mutation.
    Event $E$ is illustrated in Figure \ref{fig:eventE}.
    \item[(ii)] Event $F$ is defined to be that $E$ occurred and further that during one of
    the iterations $r\in\{i_1,...,i_H\}$
    \mutate resulted in the membership of a single element being flipped and it was
    an element of $A^*$, i.e. it is either the case that \mutate returns $X_{r-1}\cup\{a^*\}$
    for $a^*\in A^*\setminus X_{r-1}$
    or $X_{r-1}\setminus\{a^*\}$ for $a^*\in A^*\cap X_{r-1}$.
    Event $F$ is illustrated in Figure \ref{fig:eventF}.
  \end{itemize}
  We now prove a couple of results concerning the events $E$ and $F$.
  Claim \ref{claim:probF} states that if $E$ occurs then $F$ does as well with
  high probability. Claim \ref{claim:independenceF} states that the expected value of
  $f(X_{\sigma(t)})$ is independent of event $F$.

  \begin{claim}
    $P(F|E) \geq 1-\epsilon$.
    \label{claim:probF}
  \end{claim}
  \begin{proof}
    At the beginning of any iteration
    the probability that \mutate will result in the flipping of a single element
    in $A^*$ and no other changes is
    \begin{align*}
      \rho = \sum_{x\in A^*}\left(1-\frac{1}{n}\right)^{n-1}\frac{1}{n}
      = \frac{|A^*|}{n}\left(1-\frac{1}{n}\right)^{n-1}
      \geq \frac{|A^*|}{en}.
    \end{align*}
    Therefore the probability that $\argmaxf{\omega}$ has been selected on Line
    \ref{line:selectbea} $H$ times and on
    none of those iterations did \mutate result only in the flipping of a single element
    in $A^*$ is
    \begin{align*}
      P(\neg F|E) = \left(1-\rho\right)^H \leq \left(1-\frac{|A^*|}{en}\right)^H
      \leq e^{-H|A^*|/(en)}
      \overset{a}{\leq} \epsilon
    \end{align*}
    where (a) follows from Equation \ref{eqn:boundopt}.
    Claim \ref{claim:probF} then follows.
  \end{proof}

  \begin{claim}
    \label{claim:independenceF}
    $\ex{f(X_{\sigma(t)})|F} = \ex{f(X_{\sigma(t)})}.$
  \end{claim}
  \begin{proof}
    Let $\alpha\in\mathbb{R}_{\geq 0}$.
    First it is shown that $P(E|f(X_{\sigma(t)})=\alpha)=P(E)$.
    On any iteration of \bea, the probability that $\ell$ will be incremented is
    \beaprob. Whether it increments $H$ times in the interval $(\sigma(t),t]$
    does not depend on the value of $f(X_{\sigma(t)})$, therefore
    $P(E|f(X_{\sigma(t)})=\alpha)=P(E)$.

    Next, it is shown that $P(F|f(X_{\sigma(t)})=\alpha) = P(F)$.
    If it is assumed that $E$ occurred and so $\ell$ is incremented $H$ times in
    the interval $(\sigma(t),t]$, then on each of these increments the probability
    that \mutate results only in the flipping of a single element
    in $A^*$ is
    \begin{align*}
          \sum_{x\in A^*}\left(1-\frac{1}{n}\right)^{n-1}\frac{1}{n}.
    \end{align*}
    This does not depend on the value of $f(X_{\sigma(t)})$,
    therefore $P(F|E\land f(X_{\sigma(t)})=\alpha)=P(F|E)$.

    Then we use the above two facts to see that
    $P(F|f(X_{\sigma(t)})=\alpha)=P(F|E\land f(X_{\sigma(t)})=\alpha)P(E|f(X_{\sigma(t)})=\alpha)
    =P(F|E)P(E)=P(F)$.
    Then Bayes' theorem gives that
    $P(f(X_{\sigma(t)})=\alpha|F) = P(f(X_{\sigma(t)})=\alpha)$,
    which implies the statement of Claim \ref{claim:independenceF}.
  \end{proof}

  Equation \ref{eqn:inductionbea} will be proven by breaking up runs of \bea into
  those where $E$ occurs and doesn't occur, then further breaking up runs of
  \bea where $E$ occurs into those where $F$ occurs and doesn't occur, and finally
  applying the law of total probability.
  The case of runs of \bea where $E$ does not occur is an easy result of
  Lemma \ref{lemma:replace}:
  \begin{align}
    \ex{f(X_t)|\neg E} \geq \left(1-\left(1-\frac{1-\epsilon}{|A^*|}\right)^{\min\{\omega_t, P-1\}}\right)f(A^*). \label{eqn:1}
  \end{align}

  We now consider runs of \bea where event $F$ occurs.
  If $F$ occurred, then on some iteration $r\in\{i_1,...,i_H\}$
  \mutate resulted in the membership of an $a^*\in A^*$ being flipped and no other changes.
  Now two needed claims are proven. Claim \ref{claim:existY} states that if $F$ occurs
  then the expected value of $f(X_t)$ is at least the expected value of $f(X_{r-1}\cup\{a^*\})$,
  while Claim \ref{claim:gainbea} gives a lower bound on the expected value of
  $f(X_{r-1}\cup\{a^*\})$. Together, these two claims provide a lower bound on the expected
  value of $f(X_t)$ if $F$ occurs.
  \begin{claim}
    \label{claim:existY}
    $\ex{f(X_t)|F} \geq \ex{f(X_{r-1}\cup\{a^*\})|F}$
  \end{claim}
  \begin{proof}
    If $F$ occurs, then \mutate returns either $X_{r-1}\cup\{a^*\}$ for $a^*\in A^*\setminus X_{r-1}$
    or $X_{r-1}\setminus\{a^*\}$ for $a^*\in A^*\cap X_{r-1}$. Let $F_1$ be the former
    event and let $F_2$ be the latter event.

    Consider a particular run of \bea where $F$ occurs.
    Then $|X_{r-1}\cup\{a^*\}|\leq\omega_{r-1}+1=\omega_{t-1}+1=\omega_t$ since
    $\omega$ did not increment on runs in $(\sigma(t),t)$ but did on iteration $t$.
    If $F_1$ occurs on this run, then $X_{r-1}\cup\{a^*\}$ was returned by \mutate
    on iteration $r$
    and therefore $f(X_t)\geq f(X_{r-1}\cup\{a^*\})$ by Lemma \ref{lemma:replace}.
    If $F_2$ occurs on this run, then $X_{r-1}\cup\{a^*\}=X_{r-1}$. $X_{r-1}$ was
    returned by \mutate on some run before $t$ since $X_{r-1}\in\pool$ at the end
    of iteration $r-1$. Therefore $f(X_t)\geq f(X_{r-1}\cup\{a^*\})$ by Lemma
    \ref{lemma:replace}.
    Then for any run of \ea where $F$ occurs $f(X_t)\geq f(X_{r-1}\cup\{a^*\})$
    and so Claim \ref{claim:atleast} follows.
  \end{proof}
  \begin{claim}
    \label{claim:gainbea}
    $\ex{f(X_{r-1}\cup\{a^*\})|F}\geq\left(1-1/|A^*|\right)\ex{f(X_{r-1})|F}+1/|A^*|f(A^*)$
  \end{claim}
  \begin{proof}
    Consider all runs of \bea where $F$ occurs. Then it can be seen by inspecting
    \mutate that any element of $A^*$ is equally likely to be the one flipped.
    Therefore
    \begin{align*}
       \ex{\delt{X_{r-1}}{a^*}|F} &= \frac{1}{|A^*|}\sum_{y\in A^*}\ex{\delt{X_{r-1}}{y}} \\
       &\overset{a}{\geq} \frac{1}{|A^*|}(f(A^*)-\ex{f(X_{r-1})})
    \end{align*}
    where (a) follows from the monotonicity and submodularity of $f$.
    Claim \ref{claim:gainbea} follows by re-arranging the equation.
  \end{proof}

  \noindent Now we can apply the above claims in order to see that
  \begin{align}
    \ex{f(X_t)|F}
    &\overset{a}{\geq} \left(1-\frac{1}{|A^*|}\right)\ex{f(X_{r-1})|F}+\frac{1}{|A^*|}f(A^*) \nonumber \\
    &\overset{b}{\geq} \left(1-\frac{1}{|A^*|}\right)\ex{f(X_{\sigma(t)})|F} + \frac{1}{|A^*|}f(A^*) \nonumber \\
    &\overset{c}{=} \left(1-\frac{1}{|A^*|}\right)\ex{f(X_{\sigma(t)})} + \frac{1}{|A^*|}f(A^*) \label{eqn:3}
  \end{align}
  where (a) is applying Claims \ref{claim:existY} and \ref{claim:gainbea};
  (b) is applying Claim \ref{claim:improve};
  and (c) is applying Claim \ref{claim:independenceF}.
  Finally
  \begin{align}
    f(A^*) - \ex{f(X_t)|E} &= f(A^*) - \left(P(F|E)\ex{f(X_t)|F} + P(\neg F|E)\ex{f(X_t)|\neg F}\right) \nonumber \\
    &\overset{a}{\leq} f(A^*) - \left(P(F|E)\ex{f(X_t)|F} + P(\neg F|E)\ex{f(X_{\sigma(t)})}\right) \nonumber \\
    &\overset{b}{\leq} \left(1-\frac{P(F|E)}{|A^*|}\right)\left(f(A^*)-\ex{f(X_{\sigma(t)})}\right) \nonumber \\
    &\overset{c}{\leq} \left(1-\frac{1-\epsilon}{|A^*|}\right)\left(f(A^*)-\ex{f(X_{\sigma(t)})}\right) \nonumber \\
    &\overset{d}{\leq} \left(1-\frac{1-\epsilon}{|A^*|}\right)^{\omega_{\sigma(t)}+1}f(A^*) \nonumber \\
    &\overset{d}{=} \left(1-\frac{1-\epsilon}{|A^*|}\right)^{\omega_t}f(A^*) \label{eqn:4}
  \end{align}
  where (a) is because $\ex{f(X_t)|\neg F}\geq \ex{f(X_{\sigma(t)}|\neg F)}$ by Claim \ref{claim:improve}
  and $\ex{f(X_{\sigma(t)})|\neg F}=\ex{f(X_{\sigma(t)})}$ by Claim \ref{claim:independenceF};
  (b) is using Equation \ref{eqn:3}; (c) is using Claim \ref{claim:probF}; (d) is using the
  inductive assumption; and (e) is because if $E$ occurred then $\omega_{\sigma(t)}+1=\omega_t$.
  Finally, the inductive step follows by Equations \ref{eqn:1}, \ref{eqn:4}
  and the law of total probability. Therefore Lemma \ref{lemma:expectationbea} is proven.
\end{proof}

\begin{paragraph}{Lemma \ref{lemma:bernoullibea}}
  \it{
  Consider a run of \bea as a series of Bernoulli trials $Y_1,...,Y_T$,
  where each iteration is a
  trial and a success is defined to be when $\ell$ is incremented.
  $Y_1,...,Y_T$ are independent, identically distributed Bernoulli trials where the
  probability of success is \beaprob.
  }
\end{paragraph}
\begin{proof}
  In order for $\ell$ to be incremented, the element in
  $\mathcal{S}$ of cardinality $\beta$ must be chosen by \selectbea,
  which one can see by inspecting \selectbea that this occurs with probability $p$.
\end{proof}

\begin{paragraph}{End of Proof of Theorem \ref{theorem:bea}}
    By Lemma \ref{lemma:bernoullibea}, we are able to apply Chernoff's bound
    (Lemma \ref{lemma:chernoff}) to $Y_1,...,Y_T$.
    By seeing that $\ex{\sum_{i=1}^TY_i}=T$\beaprob,
    \begin{align*}
      P\left(\sum_{i=1}^TY_i < H\kappa\right)
      &\overset{a}{\leq} P\left(\sum_{i=1}^TY_i < \frac{Tp}{2\maxpointer}\right) \\
      &\overset{b}{\leq} e^{-Tp/(8\maxpointer)} \\
      &\overset{c}{\leq} \epsilon
    \end{align*}
    where (a) is because $T\geq\beaTa$;
    (b) is applying Lemma \ref{lemma:chernoff} with $\eta=1/2$;
    and (c) is because $T\geq\beaTb$.
\end{paragraph}

\subsection{$\kappa$-\beal (\kappabea)}
\label{appendix:kappabea}
\begin{algorithm}
   \caption{\kappabea$(f, \kappa, P, T, p, \epsilon)$: $\kappa$-\beal}
   \label{algorithm:kappabea}
   \begin{algorithmic}[1]
     \STATE {\bfseries Input:} \definef; $\kappa\leq n$; \defineP; \defineT; \definep;
     \defineepsilon.
     \STATE {\bfseries Output:} \definepool.
     \STATE $\beta\gets 0$, $\ell\gets 0$, $H\gets en\ln(1/\epsilon)/\kappa$
     \STATE $\pool \gets \{\emptyset\}$
     \FOR {$t\gets 1$ to $T$}
      \STATE $i\gets\uniform(\{0,...,P-1\})$
      \IF {$\flip(p)=$ heads}
       \STATE $i\gets |\argmaxf{\beta}|$\label{line:randomkappabea}
       \STATE $\ell \gets \ell + 1$
       \IF{$\ell = H$}
        \STATE $\ell\gets 0$
        \STATE $\beta\gets\beta+1$
       \ENDIF
      \ENDIF
      \IF{$\Sino{i}$ exists}
       \STATE $B\gets \Sino{i}$
       \STATE $B' \gets $ \mutate($B$)
       \IF {$|B'|<P$ and $\nexists Y \in \pool$ such that $B' \preceq Y$}
        \STATE $\pool \gets \pool \cup \{B'\}\setminus \{Y\in\pool: Y\prec B'\}$
       \ENDIF
      \ENDIF
     \ENDFOR
     \STATE \textbf{return} $\mathcal{S}$
\end{algorithmic}
\end{algorithm}




\end{document}